\documentclass[12pt, draftclsnofoot, onecolumn]{IEEEtran}

\usepackage{cite} 
\usepackage{graphicx}
\usepackage{etoolbox}
\usepackage{dtsc-creafig}
\usepackage{notationV12}
\usepackage{epsfig,epstopdf}
\usepackage{subfigure}
\usepackage{url}
\usepackage{amssymb,amsmath,amsbsy,amsfonts}
\usepackage{cancel}
\usepackage{expdlist}
\usepackage{color}
\usepackage{rotating}
\usepackage{multirow} 
\usepackage{tabularx}
\usepackage{pdflscape}
\usepackage{booktabs}
\usepackage{bmpsize}
\usepackage{pgfplots}
\usepackage{mathrsfs}
\usepackage{subfigure}
\usepackage{soul}
\DeclareGraphicsExtensions{.jpg,.pdf,.png,.gif,.eps}

\DeclareFontFamily{OT1}{pzc}{}
\DeclareFontShape{OT1}{pzc}{m}{it}{<-> s * [0.900] pzcmi7t}{}
\DeclareMathAlphabet{\mathpzc}{OT1}{pzc}{m}{it}

\graphicspath{{figures/}}

\usepackage{pgfplots}

\usepackage[withpage, printonlyused]{acronym}
\acrodef{AD}{arbitrary deployment}
\acrodef{aDEP}{approximate DEP}
\acrodef{aDEP-GS}{approximate DEP based on GS}
\acrodef{aMMSE-GS}{approximate MMSE based on GS}
\acrodef{aDEP-NSE}{approximate DEP based on NSE}
\acrodef{aDEP-CG}{approximate DEP based on CG}
\acrodef{AN}{atomic norm}
\acrodef{APP}{a posteriori probabilities}
\acrodef{AoA}{angle of arrival}
\acrodef{AoD}{angle of departure}
\acrodef{AWGN}{additive white Gaussian noise}
\acrodef{BEP}{block expectation propagation}
\acrodef{BER}{bit error rate}
\acrodef{BF}{beamforming}
\acrodef{BP}{belief propagation}
\acrodef{BP-EP}{belief propagation expectation propagation}
\acrodef{CG}{conjugate gradient}
\acrodef{CRLB}{Cram\'er-Rao lower bound}
\acrodef{CSI}{channel state information}
\acrodef{CHEMP}{channel hardening-exploiting message passing}
\acrodef{DEP}{\textit{double}-EP}
\acrodef{DEP-GS}{approximate DEP based on GS}
\acrodef{DEP-NSE}{approximate DEP based on NSE}
\acrodef{DEP-CG}{approximate DEP based on CG}
\acrodef{DFE}{decision feedback equalization}
\acrodef{EC}{expectation consistency}
\acrodef{ECC}{error correction code}
\acrodef{EMB}{enhanced mobile broadband}
\acrodef{EP}{expectation propagation}
\acrodef{EP-NSE}{approximate EP based on NSE}
\acrodef{ESPRIT}{estimation of signal parameters via rotational invariance techniques}	
\acrodef{EXIT}{extrinsic information transfer}
\acrodef{FEC}{forward error correction}
\acrodef{FIR}{finite impulse response}
\acrodef{GMP}{Gaussian message passing} 
\acrodef{GS}{Gauss-Seidel}
\acrodef{GTA}{Gaussian tree approximation}
\acrodef{IDD}{iterative detection and decoding}
\acrodef{i.i.d.}{independent and identically distributed} 
\acrodef{ICI}{intercarrier interference}
\acrodef{ISI}{intersymbol interference}
\acrodef{KL}{Kullback-Leibler} 
\acrodef{KSEP}{Kalman smoothing expectation propagation}
\acrodef{LDPC}{low-density parity-check}
\acrodef{LTI}{linear time invariant}
\acrodef{LMMSE}{linear minimum mean square error}
\acrodef{LMMSE-GS}{approximate LMMSE based on GS}
\acrodef{LMMSE-NSE}{approximate LMMSE based on NSE}
\acrodef{LLRs}{log-likelihood ratios}
\acrodef{LLR}{log-likelihood ratio}
\acrodef{LTI}{linear time invariant}
\acrodef{MAP}{maximum a posteriori}
\acrodef{MCMC}{Markov chain Monte Carlo}
\acrodef{ML}{maximum likelihood}
\acrodef{MLT}{multi--level Toeplitz}
\acrodef{MIMO}{multiple-input multiple-output}
\acrodef{MMSE}{minimum mean square error}
\acrodef{MMT}{massive machine type}
\acrodef{MSE}{mean-square-error}
\acrodef{MU-MIMO}{multi-user MIMO}
\acrodef{MUSIC}{multiple signal classification}
\acrodef{NLS}{nonlinear least squares}
\acrodef{NSE}{Neumann series expansion}
\acrodef{nuBEP}{non-uniform block expectation propagation}
\acrodef{pmf}{probability mass functions}
\acrodef{pdf}{probability density function}
\acrodef{PSD}{positive semidefinite}
\acrodef{SD}{sphere decoding}
\acrodef{SDP}{semidefinite programming}
\acrodef{SEP}{smoothing expectation propagation}
\acrodef{SER}{symbol error rate}
\acrodef{SISO}{single-input single-output}
\acrodef{SNR}{signal-to-noise ratio}
\acrodef{SPA}{sum-product algorithm} 
\acrodef{S/P}{serial to paralell} 
\acrodef{TBEP}{turbo block expectation propagation}
\acrodef{TLMMSE}{turbo linear minimum mean square error}
\acrodef{URLL}{ultra-reliable low-latency}
\acrodef{ZF}{zero forcing}

\makeatletter
\def\addlegendimage{\pgfplots@addlegendimage}
\makeatother

\newtheorem{theorem}{Theorem}
\newtheorem{lemma}{Lemma}
\newtheorem{example}{Example}

\newtheorem{corollary}{Corollary}
\newtheorem{remark}{Remark}
\newtheorem{definition}{Definition}
\newtheorem{assumption}{Assumption}
\newtheorem{conjecture}{Conjecture}

\newcommand{\qed}{\nobreak \ifvmode \relax \else
      \ifdim\lastskip<1.5em \hskip-\lastskip
      \hskip1.5em plus0em minus0.5em \fi \nobreak
      \vrule height0.75em width0.5em depth0.25em\fi}

\def\dd{d}

\def\ttt{t}
\def\qq{q}

\def\TT{{\boldsymbol{\mathsf T}}}
\def\QQ{{\boldsymbol{\mathsf Q}}}
\def\well{{\em well structured~}}
\def\mie{{\mathbf e}}
\def\miu{{\mathbf u}}
\def\miy{{\mathbf y}}

\def\mib{{\mathbf b}}
\def\mib{{\mathbf b}}
\def\mic{{\mathbf c}}
\def\mid{{\mathbf d}}
\def\mig{{\mathbf g}}
\def\mid{{\mathbf d}}
\def\mip{{\mathbf p}}

\def\mis{{\mathbf s}}

\def\mix{{\mathbf x}}
\def\miv{{\mathbf v}}

\def\j{{\mathtt j}}

\def\a{{\mathsf a}}
\def\b{{\mathsf b}}
\def\c{{\mathsf c}}

\def\m{{\mathsf \Delta}}

\def\x{{\mathsf x}}
\def\y{{\mathsf y}}
\def\z{{\mathsf z}}
\def\g{{\mathsf g}}

\def\nn{{\boldsymbol{\mathsf n}}}

\def\XX{{\boldsymbol{\mathsf X}}}
\def\YY{{\boldsymbol{\mathsf Y}}}
\def\ZZ{{\boldsymbol{\mathsf Z}}}
\def\Xsf{{{\mathsf X}}}
\def\Ysf{{{\mathsf Y}}}
\def\Zsf{{{\mathsf Z}}}

\def\Csf{{{\mathsf c}}}

\def\NNbar{{\bar{\boldsymbol{\mathsf N}}}}
\def\NN{{\boldsymbol{\mathsf N}}}

\def\mif{{\boldsymbol{\mathsf f}}}
\def\mig{{\boldsymbol{\mathsf g}}}

\def\mir{{\mathbf r}}

\def\miv{{\mathbf v}}
\def\miw{{\mathbf w}}

\def\mik{{\mathbf k}}

\def\mir{{\mathbf r}}
\def\miA{{\mathbf A}}
\def\miG{{\mathbf G}}
\def\miB{{\mathbf B}}

\def\miD{{\mathbf D}}
\def\miE{{\mathbf E}}
\def\miK{{\mathbf K}}

\def\miP{{\mathbf P}}
\def\miO{{\mathbf O}}
\def\miR{{\mathbf R}}
\def\miS{{\mathbf S}}
\def\miX{{\mathbf X}}
\def\miQ{{\mathbf Q}}
\def\miC{{\mathbf C}}
\def\miI{{\mathbf I}}
\def\miU{{\mathbf U}}
\def\miT{{\mathbf T}}
\def\miV{{\mathbf V}}

\def\miZ{{\mathbf Z}}
\def\miY{{\mathbf Y}}
\def\C{{\mathbb C}}
\def\Z{{\mathbb Z}}
\def\R{{\mathbb R}}
\def\T{{\mathbb T}}
\def\miXi{\boldsymbol{\Xi}}

\def\snr{\mathsf{SNR}}
\def\E{{\mathbb E}}
\def\H{\dag}

\def\error{{\footnotesize\frac{1}{d}\E\left\{\|\mif_k^\circ-\mif_k^\star\|_1\right\}}}
\newcommand{\LUno}[0]{$\ell_1$-AN}
\newcommand{\LZero}[0]{$\ell_0$-AN}
\newcommand{\LUnoDos}[0]{$\ell_2+\ell_1$-AN}
\newcommand{\LemmaUno}[0]{Lm. \ref{PrettyLemma}}

\DeclareMathOperator{\tr}{Tr}
\DeclareMathOperator{\diag}{diag}
\DeclareMathOperator{\rank}{rank}

\usepackage{soul}

\long\def\symbolfootnote[#1]#2{\begingroup%
\def\thefootnote{\fnsymbol{footnote}}\footnote[#1]{#2}\endgroup}

\makeatletter
\newenvironment{opteq}{%
	\let\c@equation\c@defcounter
	\renewcommand\theequation{P.\arabic{equation}}
	\begin{equation}}
{\end{equation}}

\makeatother

\newcommand*{\QEDwhite}{\hfill\ensuremath{\square}}

\usepackage{algorithm}
\usepackage{algorithmic}

	\renewcommand{\baselinestretch}{1.48}

\IEEEoverridecommandlockouts

\begin{document}


\title{Gridless Multidimensional Angle of Arrival Estimation for Arbitrary 3D Antenna Arrays}


\author{
\vspace{-0.1cm}
\authorblockN{Matilde~S\'{a}nchez-Fern\'{a}ndez, \IEEEmembership{Senior Member, IEEE}}, \authorblockN{Vahid Jamali,  \IEEEmembership{Member, IEEE}} \and \authorblockN{Jaime Llorca,  \IEEEmembership{Senior Member, IEEE}}, \and \authorblockN{Antonia Tulino, \IEEEmembership{Fellow, IEEE}}\thanks{M. S\'anchez-Fern\'andez is with Universidad Carlos III de Madrid, Spain. V. Jamali is with University of Erlangen-Nuremberg, Germany. A. Tulino, and J. Llorca  are with New York University, NY. A. Tulino is with University of Naples Federico II, Italy.}
}

\maketitle
\vspace{-2cm}
\begin{abstract}

A full multi--dimensional characterization of the \ac{AoA} has immediate applications to the efficient operation of modern wireless communication systems.
In this work, we develop a compressed sensing based method to extract multi-dimensional \ac{AoA} information exploiting the sparse nature of the signal received by a sensor array. The proposed solution, based on the atomic $\ell_0$ norm, enables accurate gridless resolution of the AoA in systems with arbitrary 3D antenna arrays. Our approach allows characterizing the maximum number of distinct sources (or scatters) that can be identified for a given number of antennas and array geometry. Both noiseless and noisy measurement scenarios are addressed, deriving and evaluating the resolvability of the AoA propagation parameters through a multi--level Toeplitz matrix $\rank$--minimization problem. To facilitate the implementation of the proposed solution, we also present a least squares approach regularized by a convex relaxation of the $\rank$-minimization problem and characterize its conditions for resolvability. 
\end{abstract} 
\vspace{-0.5cm}
\begin{IEEEkeywords}
\vspace{-0.3cm}
Multi--dimensional AoA, atomic norm, $3$D antenna array, multi--level Vandermonde decomposition.
\end{IEEEkeywords} 
\acresetall 
\vspace{-0.5cm}
\section{Introduction}
Optimizing the performance of wireless communication systems critically relies on extracting relevant information from the wireless propagation channel. Array processing techniques have been extensively used to reveal, extract, and exploit key propagation parameters such as \ac{AoA} and \ac{AoD} in order to tailor and optimize the transmission to individual users in systems with multiple antennas \cite{Heath16}. 

In particular, \ac{AoA} (and analogously \ac{AoD}) estimation refers to the process of retrieving multi--dimensional direction of propagation parameters (e.g., azimuth, elevation) associated with possibly multiple electromagnetic sources (e.g., transmitters, scatters), from the observation of a receiving vector $\miy$, whose elements represent the signal outputs at each of the receiving antennas that form a sensor array. 
The received signal associated with a given source $k\in\{1,\ldots,K\}$ impinging on an $N$-antenna array from direction $(\theta_k, \phi_k)$ can be represented by a so-called {\em steering vector} $\mir_N\left(\theta_k, \phi_k\right)\in\C^N$, characterized by the relative phase shifts at each of the antennas. The observable $\miy$ is hence a  noisy linear combination of such steering vectors, one for each source. The goal is to identify each steering vector and extract their associated multi--dimensional AoA parameters, represented by multi--dimensional {\em frequencies},  from the observation of $\miy$.

 Existing methods for angular characterization of signal propagation include initial beamforming--based approaches \cite{Capon69} (highly limited in their resolution), covariance-based subspace methods such as Capon’s beamformer, \ac{MUSIC}, \ac{ESPRIT}, and their extensions \cite{Stoica89,Roy89}, and (determistic) \ac{ML} methods based on \ac{NLS} optimization problems \cite{Stoica89}. However, all aforementioned methods suffer from well--known limitations. Subspace and \ac{NLS} based methods require a priori knowledge of the number of sources, which may be difficult to obtain in practice. Capon’s beamformer, \ac{MUSIC}, and \ac{ESPRIT} need accurate estimates of the observable covariance matrix, which in turn requires a slow time-varying environment, a large number of observations (to build the sample covariance matrix), and a low level of source correlation (to avoid rank deficiency in the sample covariance matrix). In addition, \ac{NLS} algorithms require very accurate initialization since their objective function has a complicated multimodal  shape with a sharp global minimum.

On the other hand, an important observation is that \ac{AoA} estimation can be seen as a sparse data representation/separation problem, where {\em compressed sensing} methods \cite{Candes06} can be used to significantly enhance accuracy and robustness. In particular, compressed sensing does not require any statistical knowledge of the parameters to be estimated or associated observations. Hence, it avoids the need to estimate the observable covariance matrix or even having prior knowledge of the number of sources to identify, while still providing strong performance guarantees. Such obliviousness to the number of sources and the statistics of the received signals makes compressed sensing techniques notably suitable to overcome the limitations of previous approaches in terms of robustness to mismatches with respect to prior knowledge.

The estimation of \ac{AoA} information via compressed sensing techniques can be classified into three main categories: on--grid, off--grid, and gridless estimation. Despite \ac{AoA} being an intrinsically continuous variable, initial approaches for \ac{AoA} estimation using sparse methods were based on sampling the angular space, forcing the angles to be estimated to lay on a discrete grid \cite{Stoica11,Heckel18}. The problem with this approach is clear -- even with very fine gridding, there is always a mismatch between grid and real values. Off--grid solutions try to overcome this mismatch by means of adaptive grids \cite{Hu12} or by estimating and compensating the grid offset \cite{Yang13}. Finally, gridless approaches do not rely on any discretization of the angular space -- they directly operate in the continuous domain, completely avoiding the grid mismatch problem,  at the expense of increased computational complexity \cite{Bhaskar13,Tang13}. 

A number of recent works have studied the use of gridless compressed sensing techniques for AoA estimation. However, most of them have only been applied to the estimation of $1$D \ac{AoA} parameters in \cite{Bhaskar13,Tang13,Tan14,Semper18}, and only more recently, to multidimensional scenarios, with applications to uniform antenna arrays in \cite{Yang16,Chu19}, and allowing non-uniform deployments in \cite{Chi15,Pan18}. 

It is important to note that non--uniform $3$D deployments are becoming increasingly relevant in upcoming $5$G/$6$G communication paradigms such as \ac{EMB}, \ac{MMT}, or \ac{URLL} communications, where high rates, low latency, massive connectivity, and extremely accurate localization 
require a level of multi--dimensional space awareness that can only be achieved with arrays of matching dimensionality. Importantly, $3$D arrays rarely deploy antennas uniformly,  since they typically follow cubic \cite{Zheng12} or cylindrical \cite{Wu17b} geometries without antennas deployed inside the volume. 
\vspace{-0.5cm}

\subsection{Contributions}
\label{Sec:Contributions}

In this work, we focus on the design of robust an efficient techniques for \emph{i)} full--dimensional \ac{AoA} estimation in systems with \emph{ii)} arbitrary antenna deployments. We leverage compressed sensing techniques that exploit the sparse nature of signal measurements impinging on 
arbitrary multi--dimensional antenna arrays. The approach undertaken 
is based on sparse approximation of the signal impinging on a $3$D antenna array and the use of the Vandermonde decomposition to enable gridless extraction of multi--dimensional propagation parameters that are intrinsically continuous variables (see \cite{Tang13,Tan14,Yang16}). It is worth highlighting that while the application scenario considered in this paper focuses on $3$D array deployments where the goal is  the gridless extraction of $d$-dimensional parameters with $d \leq 3$,  the results derived in Section \ref{sec_AN} hold for  multi--dimensional parameters of arbitrary dimensionality $d$.

Our main contributions can be summarized as follows:

In the noiseless setting: 
\begin{itemize}

\item 
Given a receiving $N$-antenna array following an arbitrary $d$-dimensional configuration, we characterize the \emph{resolvable region}, i.e., the maximum number of resolvable sources (e.g., scatters) $K$ and their associated $d$-dimensional propagation parameters, referred to as {\em frequencies}. We show that the size of the resolvable region grows linearly with the number of antennas, not only in the uniform setting, but also with non--uniform arrays. Specifically, the maximum number of resolvable scatters is given by $\left \lfloor \frac{S_\Csf-(d-1)}{2} \right \rfloor$, where $S_\Csf$ is the sum of the number of antennas in each dimension of the largest uniform array embedded in the original structure. We further strongly conjecture, as validated by our numerical results, 
that the maximum number of resolvable scatters is in fact given by $\left \lceil \frac{N_\Csf}{2}-1\right \rceil$, where $N_\Csf$ is the total number of antennas in the largest uniform array embedded in the original structure. Importantly, the resolvable region characterized in this work \emph{i)} is larger than the region already available in literature \cite{Yang16} and \emph{ii)} has more general applicability, since it can be used with measurements coming from non-uniform sampled observations.

\item 
For the case of non-uniform deployments with large number of antennas (e.g., massive \ac{MIMO} systems), we show that the large-scale nature of the array deployment allows guaranteeing resolvability with probability $1-\epsilon$ of up to $K$ scatters with $N=O(K\log(K/\epsilon))$  independently of the specific array geometry. 


\item  
The aforementioned characterization follows a novel constructive approach that allows not only identifies the resolvable scatters, but also extracts the associated $d$-dimensional propagation parameters. Such approach consists of three steps: First, the $N$-dimensional observable vector $\miy$ is represented as a sampled version of an enlarged $\bar N$-dimensional vector $\mis$ composed of a linear combination of {\em uniform} steering vectors associated with the smallest 
uniform array covering each of the original antennas, referred to as {\em virtual uniform array}. In the second step, we find the minimum number of steering vectors that compose $\mis$ by reformulating the minimization of the $\ell_0$-\ac{AN} of $\mis$ as the search for a $\min$-$\rank$ $d$-dimensional \ac{MLT} matrix. A distinctive feature of our approach is to restrict this search to a specific \ac{PSD} matrix set, termed {\em canonical} (see Definition \ref{definition2}), that is shown to preserve optimality, while significantly reducing the required dimension of the virtual uniform array $\bar N$. Finally, in the third step, we show that the multilevel Vandermonde decomposition of the resulting $\min$-$\rank$ $d$-\ac{MLT} canonical matrix allows the extraction of the $d$-dimensional propagation parameters. 


\item 
The proposed $\rank$ minimization based approach is NP-hard in general, and its traditional convex counterpart based on $\ell_1$-\ac{AN} and trace minimization is known to exhibit significant limitations \cite{Tang13}. Most prominently, the resulting $\ell_1$-\ac{AN} and trace minimizer is not guaranteed to be sparse. To circumvent this problem, we resort to  a weighted $\ell_2$+$\ell_1$-\ac{AN} minimization method. We numerically show the advantage of the proposed approach in overcoming the drawbacks of the $\ell_1$-\ac{AN} convex approximation.

\item
Noting that the complexity of the proposed $\rank$ and trace minimization based methods is driven by the size of the virtual uniform array, our approach allows finding the $d$-dimensional propagation parameters of  any number of scatters in the resolvable region $K$ using a virtual uniform array with at least $K$ antennas along its largest dimension, leading to a complexity that scales (exponentially for $\ell_0$-\ac{AN} and polynomially for $\ell_1$-\ac{AN}) with $K$. This result significantly outperforms the existing best known result in \cite{Yang16} that requires a virtual uniform array with at least $K$ antennas along its smallest dimension, resulting in a complexity that scales (exponentially or polynomially)
with $K^\dd$.
The reduced requirement on the dimension of the virtual uniform array is a consequence of a novel generalization of the Carathéodory-Fejér theorem, which allows enlarging the set of \ac{PSD} $d$-\ac{MLT} matrices  that  admit a unique Vandermonde decomposition (see Lemma \ref{PrettyLemma}).


\end{itemize}

In the noisy setting:
\begin{itemize}
\item For the case of noisy measurements, 
we resort to the weighted $\ell_2$+$\ell_1$-\ac{AN} optimization approach 
and provide a closed-form expression for the weight as a function of the signal-to-noise ratio $\snr$, which is shown to well approximate the numerically computed optimal value. We evaluate the performance of our approach for moderate and high $\snr$, showing an estimation error very close to the associated \ac{CRLB}.   

\end{itemize}

General array design guidelines:
\begin{itemize}
\item
Our results suggest that if the goal is identifying up to $K$ scatters (and extract their associated $d$--dimensional propagation parameters), we should design an antenna array whose smallest uniform embedded structure has at least $2K$ antennas, where most of them are lying along a dominate dimension, while maintaining at least $2$ antennas in the other dimensions. While $2(K+1)$ antennas, independent of geometry, allows resolving $K$ scatters, a design with a dominant dimension will help resolving the $K$ scatters with complexity $O(K)$ instead of $O(K^\dd)$. 



\end{itemize}
\vspace{-0.5cm}
\subsection{Related work}



We now describe results and limitations of the compressed sensing based studies that are more closely related to our work. 
We emphasize that the limitations of such previous results arise from the fact of requiring more stringent conditions for frequency recovery and/or lacking applicability to general array structures.

In terms of the resolvable region, the best known result, provided in \cite{Yang16}, bounds the number of resolvable scatters $K$ by the minimum number of antennas in each of the $d$ dimensions of the antenna array. 
In contrast, we significantly enlarge the resolvable region by bounding the number of resolvable scatters by the maximum number of antennas among each dimension. In terms of the number of antennas to resolve $K$ scatters, our approach improves the result in \cite{Yang16} by requiring $O(K)$ antennas instead of $O(K^\dd)$.
In addition, 
while the recovery conditions in \cite{Yang16} are only applicable to uniform antenna deployments, 
our approach applies to arbitrary $3$D antenna arrays. 
Finally, different from \cite{Yang16}, we ensure uniqueness of the solution to the $\ell_0$ optimization problem by introducing a condition on the dimension of the measurement vector (see Theorem \ref{theoremApiProbalistic}) or imposing a specific array geometry (see Corollary \ref{CorollaryTh1}).


In terms of non-uniform array deployments, the works of \cite{Tang13}, \cite{Chi15}, and \cite{Pan18} initiated such analyses by considering non-uniform geometries as {\em random} samples of uniform  array structures in the asymptotic regime. Specifically, \cite{Tang13} focused on the one-dimensional setting (i.e., linear array) and provided recovery conditions that allow recovering $K$ scatters with high probability (i.e., 
$1-\epsilon$) given that there is a minimum separation among the frequencies and a number of active antennas $N =  O\left(K \log{ \left(K/\epsilon \right)}
\log{ \left(\bar{N}/\epsilon \right)}\right)$ with $\bar{N}$ the number of antennas in the uniform array structure. The works of \cite{Chi15} and \cite{Pan18} extend the approach in \cite{Tang13} to the $2$-dimensional setting (i.e., square planar array). While \cite{Chi15} provides recovery conditions that are exactly the extensions to the $2$-dimensional setting of the ones provided in \cite{Tang13}, \cite{Pan18} formulates a low complexity implementation of the recovery algorithm, albeit less practically verifiable recovery conditions. The major drawback of the results in \cite{Tang13}, \cite{Chi15}, and \cite{Pan18} is that for most common applications, non-uniform array deployments can not been seen as the result of sampling uniform array structures uniformly at random. Our work extends and generalizes the results in \cite{Tang13} and \cite{Chi15} considering a completely arbitrary non-uniform deployment, which can be either deterministic or arbitrarily random, it does not require a minimum separation among the frequencies, and it requires a number of active antennas that grows as $O\left( K \log(K/\epsilon)\right)$ (see Theorems \ref{theoremApiProbalistic} and \ref{theoremApiProbalisticl1}). 



The authors in \cite{Heckel18} extended the the recovery conditions stated in \cite{Tang13} and  \cite{Chi15} to a more general set of sensing matrices provided that they obey isotropy and incoherence conditions. However, differently from \cite{Tang13} and  \cite{Chi15},  
\cite{Heckel18} resorts to on-grid approaches to identify the propagation parameters or frequencies, resulting in potentially significant gridding errors. 

\textit{Notation:} $[\cdot]^\top$ is the transpose and $[\cdot]^{\H}$ is the Hermitian. The operator $\diag(\mathbf{x})$, returns a diagonal matrix with diagonal given by $\mathbf{x}$. 
Also, $\miX_{\left(m\right)}=\mix_m$ and $\miX^{\left(m\right)}=\mix^m$ are respectively the $m$-th column\,/\,row of matrix $\miX$, $x_{nm}$ is the matrix element in row $n$ and column $m$, and $\miX^{(\mathcal{I})}$ and  $\mix^{(\mathcal{I})}$ are the submatrix of $\miX$  and the subvector of $\mix$ given respectively by the rows and elements in the index set $\mathcal{I}$. For a given integer $K\in\Z$, $\left[K\right]=\left\{1,\dots,K\right\}$. $\T$ denotes the unit circle $\left[0, 1\right]$ by identifying the beginning and the ending points.  $\|\cdot\|_\rho$ represents the $\ell_\rho$ norm and $\left \| \miX \right \|_{2 \rightarrow 2}$  represents the squared-root of the largest eigenvalue of the Hermitian matrix $\miX^\dagger\miX$.  The $y$--modulus of value $x$ is given by $\mod\left(x,y\right)$. A $N\times M$ all-zeros matrix is given by $\mathbf{0}_{N\times M}$ and $\mathbf{1}_{n}$ is the all--ones vector of dimension $n\times 1$. The different products of vectors and matrices are the inner product of vectors $\mix$ and $\miy$ represented with $\mix\cdot\miy$, the Kronecker product represented by $\otimes$ and the Khatri--Rao product represented with $\odot$.  $\miA^\g$ denotes the generalized inverse of matrix $\miA$;  and $|c|$ and $\angle c$ are the absolute value and angle of complex number $c$ in polar coordinates, respectively.

\section{Characterization of the steering vectors in the unit circle frequency domain}\label{sec:II}

\subsection{General model for steering vectors}

Let us consider a propagation scenario with a receiving $N$-antenna array, partly characterized by the response of the steering vector $\mir_N\left(\theta,\phi\right)\in\C^N$ in the direction $(\theta,\phi)$, identified as the azimuth and elevation angles, respectivelly $\theta\in\left[0,2\pi\right]$ and $\phi\in\left[-\frac{\pi}{2},\frac{\pi}{2}\right]$. The structure of the measuring steering vector depends on the number of antennas $N$ and on their normalized relative positions in the sensor array $\mip_n=[p_n^x, p_n^y,p_n^z]^\top$ with $n\in[N]$. The phase at the $n$-th antenna is given by $\Phi_n=-{2\pi}\big(p_n^x\sin\theta\cos\phi+ p_n^y\sin\theta\sin\phi +p_n^z\cos\theta \big)$. Then, the measuring steering vector can be represented as 
	\begin{IEEEeqnarray}{lll}
		\mir_N(\theta,\phi)&=\frac{1}{\sqrt{N}}\big[e^{\j\Phi_1},   \ldots, e^{\j\Phi_N}\big]^\H.
		\label{eq:rd}
	\end{IEEEeqnarray}
In the following subsection, we first simplify the steering vector in \eqref{eq:rd} for uniform arrays. Subsequently, we show that the steering vector of a general non-uniform array can be equivalently represented by a sampled version of the steering vector of a virtual uniform array via a proper sampling (or sensing) matrix.  

\subsection{Steering vector for a uniform $3$D array}
\label{SV_3DUL}

 Consider a uniform $3$D array deployment ($3$D-UD) with $N=\Xsf\Ysf\Zsf$ antenna elements, where $\Xsf$,  $\Ysf$, and $\Zsf$ are the number of elements deployed in each spatial dimension of the Cartesian coordinate system. We define the $d=3$ dimensional vector  $\NN\triangleq\left[\Xsf,\Ysf,\Zsf\right]$ identifying the number of antenna elements, and deploy the antennas with normalized spacing $\delta^x$, $\delta^y$ and $\delta^z$ in each of the dimensions. The $n$-th antenna normalized position is given in this case by $\mip_n=\left[\x_n \delta^x,\y_n \delta^y,\z_n \delta^z\right]^\top$ with $\x_n\in\left\{0,\dots, \Xsf-1\right\}$, $\y_n\in\left\{0,\dots, \Ysf-1\right\}$, and $\z_n\in\left\{0,\dots,\Zsf-1\right\}$. Then, the phase simplifies to $\Phi_n=-{2\pi}\big(\x_n \delta^x\sin\theta\cos\phi +\y_n \delta^y\sin\theta\sin\phi+\z_n\delta^z\cos\theta\big)$.

Let us  define a normalized frequency vector  $\mif=\left[f^x,f^y,f^z\right]^\top\in\T^3$ that contains the information on both azimuth and elevation, with $f^x=\mod\left(\delta^x\sin\theta\cos\phi,1\right)$, $f^y=\mod\left(\delta^y\sin\theta\sin\phi,1\right)$, and $f^z=\mod\left(\delta^z\cos\theta,1\right)$ and the normalized position of antenna $n$-th as $\nn_n=\left[\x_n,\y_n,\z_n \right]^\top$. Then, we obtain the phase $\Phi_n=-2\pi\left(\x_n f^x+\y_n f^y+\z_n f^z\right)+ 2m\pi=-2\pi\mif\cdot \nn_n + 2m\pi$, where $m$ is an integer. The receive antenna steering vector $\mir_\NN\left(\mif\right)$ is obtained as  
\begin{IEEEeqnarray}{rll}\label{rNdecomp}
\mir_\NN\left(\mif\right)
=\,&\frac{1}{\sqrt{N}}\big[e^{\j 2\pi\mif\cdot \nn_1},e^{\j 2\pi\mif\cdot \nn_2},\dots,e^{\j 2\pi\mif\cdot \nn_N}\big]^\H
=\,&\mir_{\Xsf}(f^x)\otimes\mir_{\Ysf}(f^y)\otimes\mir_{\Zsf}(f^z).
\end{IEEEeqnarray}
where $\mir_\Xsf(f^x)=\frac{1}{\sqrt{\Xsf}}\big[e^{\j 2\pi \x_1f^x  },e^{\j 2\pi \x_2f^x  },\dots,e^{\j 2\pi \x_\Xsf f^x  }\big]^\H$, and $\mir_\Ysf(f^y)$, $\mir_\Zsf(f^z)$ are defined equivalently. We assume a $d$-dimensional frequency vector  $\mif\in\T^d$, where $d\in\{1,2,3\}$. In the most general case, if the antenna array is deployed along all dimensions of the Cartesian system, we have $d=3$. From \eqref{rNdecomp}, $3$D-UD antenna deployments lead to a regular structure of the steering vector which we refer to as \emph{uniform} steering vector.

  \begin{figure*}[t!]
	\centerline{\subfigure[]{\includegraphics[width=0.25\linewidth]{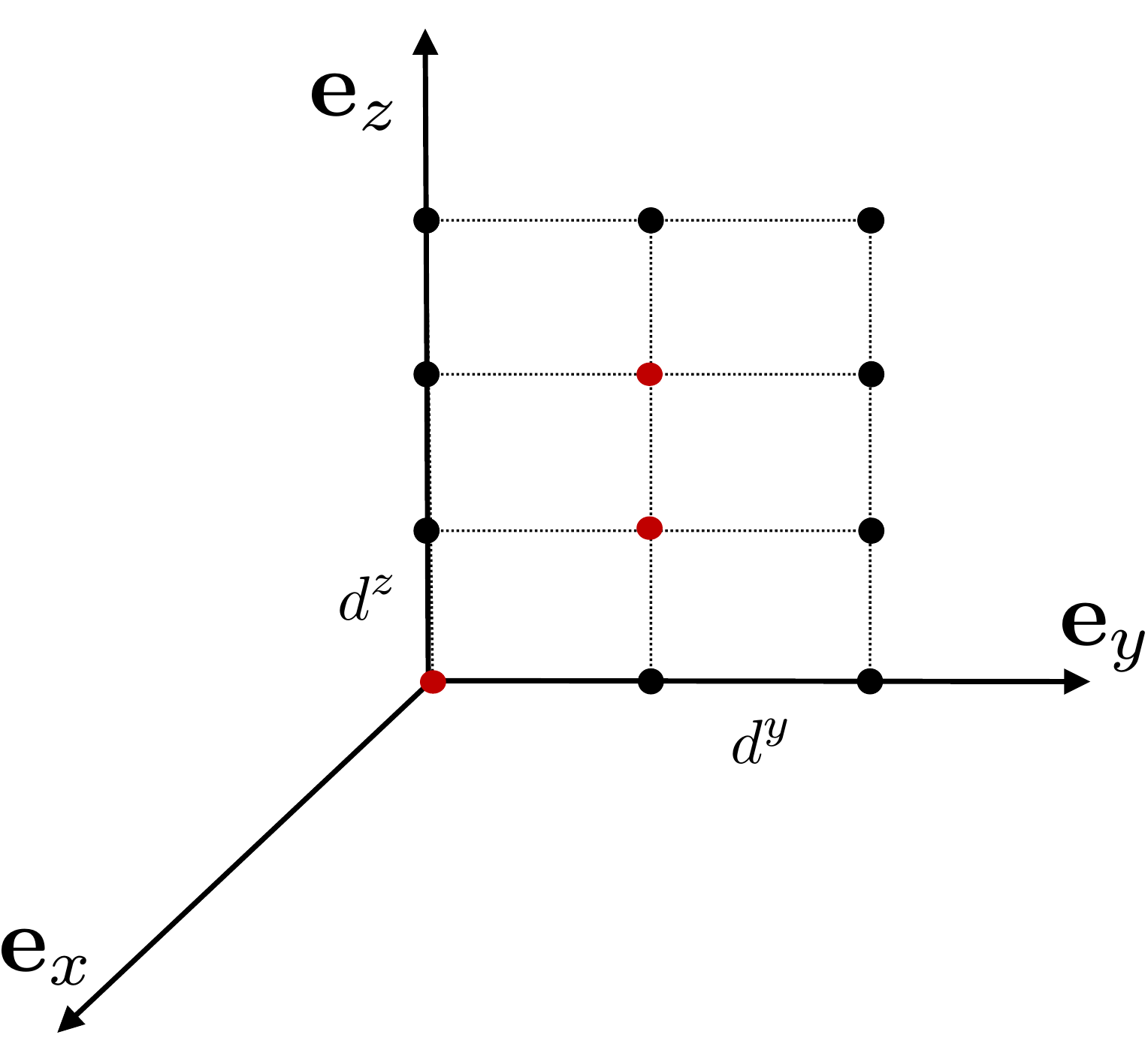}
			\LABFIG{fig:2D}}
		\hfil
		\subfigure[ ]{\includegraphics[width=0.25\linewidth]{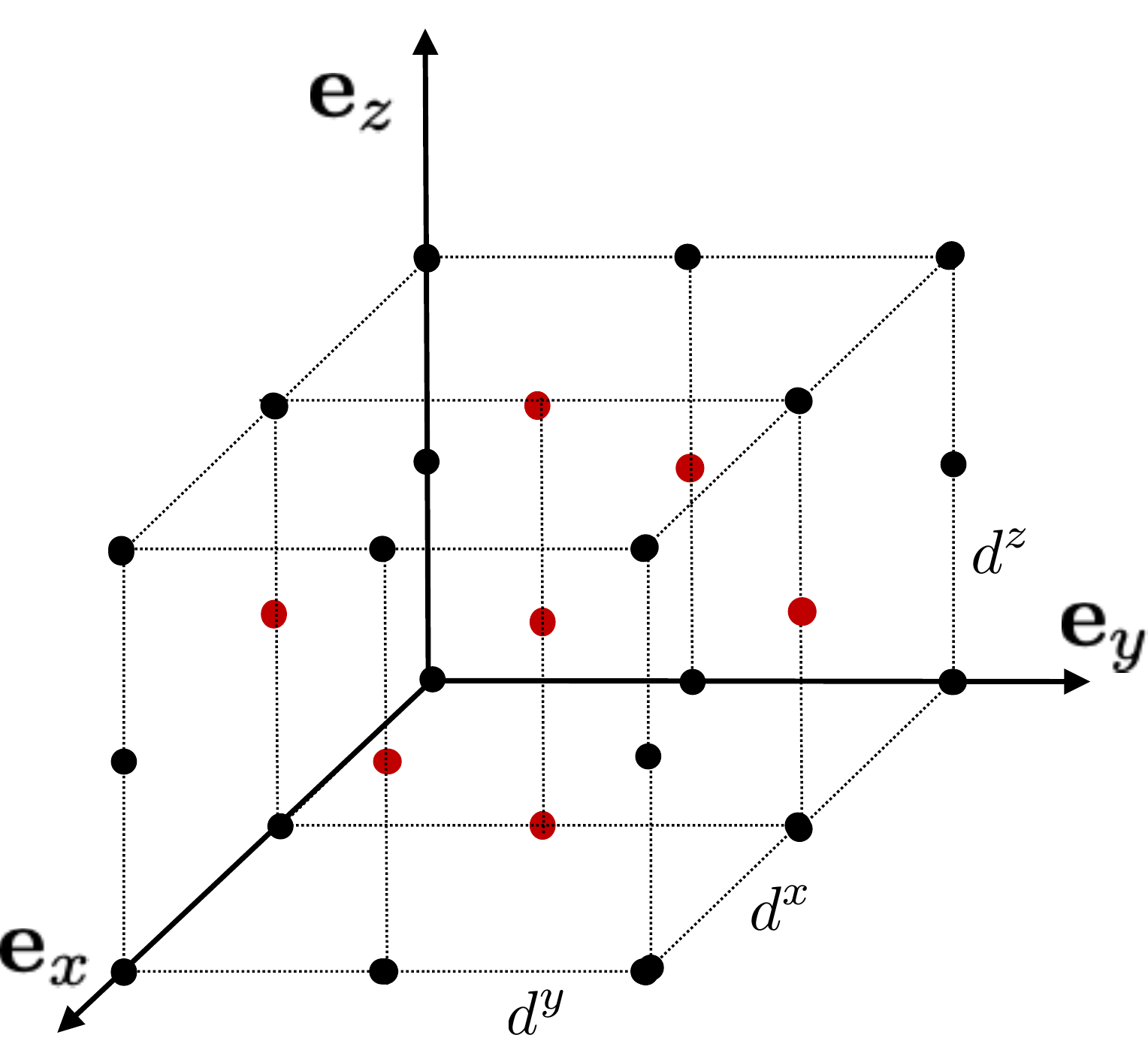}
			\LABFIG{fig:3D}}
	}
	\vspace{-0.4cm}
	\caption{Sensing an uniform array to obtain $\dd$D-\acs{AD} antenna configurations. 
	}
		\LABFIG{fig:array}
\end{figure*}

 \subsection{Arbitrary array deployments}
 \label{sec:AD}
 
Many antenna deployments, particularly those using cubic or cylindrical geometries, \cite{Zheng12, Wu17b}, cannot be described, in terms of their associated steering vectors, following the regular structural property mentioned in Section \ref{SV_3DUL}. In the following we refer to such types of  \ac{AD} as an $3$D-\ac{AD} array. 

Given any arbitrary $3$D-\ac{AD} array, the associated steering vector $\mir_{{N}}\left(\mif\right)$, even though not characterized by any regular structure, can be obtained starting from an encompassing steering vector of a virtual $3$D-UD as follows. Let ${\Xsf}\delta^x$,  ${\Ysf}\delta^y$ and 
${\Zsf}\delta^z$ be the maximum normalized antenna coordinates along the directions of the  $3$D-\ac{AD} array, where here $\delta^x$, $\delta^y$, $\delta^z$ represents the minimum normalized spacing between the array elements in each direction. Next, consider the virtual $3$D-UD  with
$\NNbar=[\bar{\Xsf},\bar{\Ysf},\bar{\Zsf}]$ elements, encompassing the $3$D-\ac{AD} array, such that
$\bar{\Xsf}\geq {\Xsf}$,  $\bar{\Ysf}\geq {\Ysf}$
and $\bar{\Zsf}\geq {\Zsf}$ and $\bar{N}=\bar{\Xsf}\bar{\Ysf} \bar{\Zsf} \geq N$, with associated steering vector $ \mir_{\NNbar}\left(\mif\right)$:
\begin{equation}
\mir_{\bar{\NN}}\left(\mif\right)= 
\mir_{\bar{\Xsf}}(f^x)\otimes\mir_{\bar{\Ysf}}(f^y)\otimes\mir_{
\bar{\Zsf}}(f^z).
\label{eq:rbar}
\end{equation}
Then, $\mir_{{N}}\left(\mif\right)$ can be obtained  starting from  $ \mir_{\NNbar}\left(\mif\right)$ using a \emph{binary sensing matrix} $\miA\in\{0,1\}^{N\times \bar{N}}$, which effectively removes some elements from $\mir_\NNbar\left(\mif\right)$ (see \FIG{fig:2D} and \FIG{fig:3D}), as:
\footnote{{ Note that $\mir_{{ N}}\left(\mif\right)$ refers to a steering vector of dimension $N$ without any specific structure, while $\mir_{\NN}\left(\mif\right)$ and $\mir_{\bar{\NN}}\left(\mif\right)$ represent uniform arrays with a deployment respectively given by $\NN=[\Xsf,\Ysf,\Zsf]$ and $\NNbar=[\bar{\Xsf},\bar{\Ysf},\bar{\Zsf}]$ elements.}}
\begin{equation}
 \mir_{N}\left(\mif\right)=\miA \mir_{\bar{\NN}}\left(\mif\right)
\label{eq:sensing_r}
\end{equation}
with $\miA$ being a fat matrix with elements $a_{ij}=0$  if component $j\in[\bar{N}]$ in the generating virtual $3$D-UD array is not included as component $i\in[N]$ in the $3$D-\ac{AD} array, and $a_{ij}=1$ otherwise. 
Note that in order to ensure that the sensing matrix removes elements from the $3$D-UD and that one particular element in the $3$D-UD is not considered more than once in the $3$D-\ac{AD} the sensing matrix $\miA\in\{0,1\}^{N\times \bar{N}}$ should meet both following conditions $\miA\mathbf{1}_{\bar{N}}=\mathbf{1}_{N}$ and $\miA^\top\mathbf{1}_{N}=\mathbf{1}_{\bar{N}}$. This is equivalent to impose that the sensing matrix $\miA\in\{0,1\}^{N\times \bar{N}}$ belongs to the {\em sensing  set} 
$\mathcal{A} \subset \{0,1\}^{N\times \bar{N}}$ defined as:
\begin{equation}
\mathcal{A}= \big\{\miA\in\{0,1\}^{N\times \bar{N}}: \, \big[\miI_N|{\bf 0}_{N\times (\bar{N}-N)}\big]  {\bf \Pi},\, {\bf \Pi} \in\mathcal{P}\big\}
\label{eq:sensing_set}
\end{equation}
with  $\mathcal{P}$ the set of all $\bar{N}\times \bar{N}$ permutation matrices, and $N\leq\bar{\Xsf}\bar{\Ysf}\bar{\Zsf}=\bar{N}$. 
Each $3$D-\ac{AD} array can be then described using the associated virtual $3$D-UD array and sensing matrix $\miA$. Note that there are multiple pairs of virtual $3$D-UD array and sensing matrix $\miA$ that
can be  associated to a given $3$D-\ac{AD} array.

 \begin{remark}
 \label{Kronorder}
 Note that the  ordering of the Kronecker product in \eqref{eq:rbar}, induces, via the $3$D-UD steering vector $\mir_{\bar{\NN}}$, an indexing of the virtual array antenna elements.  Such Kronecker ordering is arbitrary and it will be in the following properly chosen to guarantee the maximum system capability in terms of  number of resolvable directions $\mik$. In the following, to refer to the Kronecker ordering  given in   \eqref{eq:rbar} we use the notation $\bar{\Zsf}\to\bar{\Ysf}\to\bar{\Xsf}$.
\end{remark}


\section{Gridless parameter extraction}
\label{sec_AN}
Starting from an arbitrary $\dd$D-\ac{AD} array with $d \leq 3$ 
and its associated general steering vector structure described in
Sec. \ref{sec:AD}, equations \eqref{eq:rbar}-\eqref{eq:sensing_r}, in this section, we focus on the problem of
extracting key multi–dimensional propagation parameters, such  as  \ac{AoA},  from the signal received at the antennas elements of the $\dd$D-\ac{AD} array. 
Specifically, given a $\dd$D-\ac{AD} array, under the assumption of $K$ multiple sources incident  on the  array from $K$ \ac{AoA}s  (or equivalently $K$ local scatters reflecting the impinging incoming signal from $K$ \ac{AoD}s),  the received signal, measured 
at the antennas elements,  $\miy\in\C^N$, can be written as: 
\begin{IEEEeqnarray}{lll}
\miy=\miA\miR_{\bar{\NN}}\left(\mif^\star_{1:K}\right)\miu^\star+\miw=\miA\mis^\star+\miw,
\label{ec_parametrical}
\end{IEEEeqnarray} 
where $\miA\in\mathcal{A}$ and $\miR_{\bar{\NN}}\left(\mif^\star_{1:K}\right)$$=
$$\left[ \mir_{\bar{\NN}}(\mif^\star_1), \dots,  \mir_{\bar{\NN}}(\mif^\star_K)\right]\in\C^{\bar{N} \times K}$ are, respectively, the sensing matrix and the steering vector matrix of the encompassing virtual $3$D-UD associated to the $\dd$D-\ac{AD} array as defined in Sec. \ref{sec:AD}, $\miu^\star=\left[u^\star_1,\dots,u^\star_K\right]\in\C^K$ is the unknown incoming signal vector incident on the array from the $K$ \ac{AoA}s, and  $\miw \in \C^{N}$ is the \ac{AWGN} noise vector whose elements are independent with zero mean and variance $\sigma^2$. 

The objective is to retrieve from the measurement $\miy\in\C^N$, in both noiseless (i.e. $\miw=\mathbf{0}_{N\times 1}$) and noisy scenarios (i.e $\miw\neq\mathbf{0}_{N\times 1}$),
the $K$--scatter vector $\mis^\star$, the number of sources/scatters, $K$, and the associated set of frequencies, $\mif^\star_{1:K}=\left\{\mif^\star_{1},\mif^\star_{2}\dots,\mif^\star_{K}\right\}$ that characterizes the $\dd$D \ac{AoA}s of each source/scatter, with $\mif^\star_{k}\in\T^d$ and $d\leq 3$.
All our frequency recovery conditions are derived under the following assumption:
\begin{assumption}
\label{A1}
The $\dd$D frequency vectors, $\mif^{\star}_k$ with $ k\in[K]$,
associated to the  $\dd$D \ac{AoA}s of the $K$ sources/scatters are modeled as independent and identically distributed random vectors whose components are independent and uniformly distributed on $[0,1)$, i.e. $f^{\alpha\star}_k \sim \mathcal{U}[0,1)$,  for $ k\in[K]$, and $\alpha\in\{x,y,z\}$.
\end{assumption}

Note that, given our work's application scenario, this assumption is not restrictive. In fact, in wireless propagation environments, sources/scatters are typically modeled to be uniformly random placed in the surroundings of the receiver. 

Furthermore,  even though the derivations of the frequency recovery conditions are conducted for $d\leq3$ given the application scenario of our work, they are applicable also to $d>3$.

\subsection{Signal and frequency recovery in noiseless scenarios} 
\label{subsect:mavafanculo}

\subsubsection{Problem statement and previous results}
\label{sec:Problemstatement}
In this section, we focus on the noiseless setting, i.e:
\begin{equation}
\miy=\miA\miR_{\bar{\NN}}\left(\mif^\star_{1:K}\right)\miu^\star=\miA\mis^\star,
\label{ec_parametrical_noiseless}
\end{equation}
with 
$\miA\in\mathcal{A} \subset \C^{N \times\bar{N}}$. 
As already stated our objective is to retrieve $\mis^\star$, $K$, and $\mif^\star_{1:K}$.
As first step, let us provide the following definitions.

\begin{definition}

	Given a $\bar{N}$-dimensional vector $\mis$, the $\ell_0$ \ac{AN} ($\ell_0$-\ac{AN}) of $\mis$ in  $\mathcal{R}=\left\{\mir_{\bar{\NN}}(\mif):\mif\in\T^d\right\}$ is defined as $
	\left\|\mis\right\|_{\mathcal{R},0}=
	\inf_{\mif_k\in \T^d,u_k\in\C}\bigg\{ K : \, \,  \mis=\sum_{k=1}^{K} u_k\mir_{\bar{\NN}}\left(\mif_k\right)\bigg\}.$
	\QEDwhite
\end{definition}

\begin{definition}
	A $K$--scatter vector $\mis$ is defined as a vector whose atomic $\ell_0$ norm is equal to $K$. \QEDwhite
\end{definition}

It is well known that for a given measurement $\miy$ and a given sensing matrix $\miA\in\mathcal{A}$ as  in 
\eqref{eq:sensing_set}, the $K$--scatter vector $\mis^\star=\miR_{\bar{\NN}}\left(\mif^\star_{1:K}\right)\miu^\star$
can be reconstructed, from the noiseless measurement $\miy$ in \eqref{ec_parametrical_noiseless}, as the unique solution of:
	\begin{opteq}
		\underset{\mis \in \C^{\bar{N}}}{\min} \left\|\mis\right\|_{\mathcal{R},0} \quad {\rm s.t.}  \quad  \miy=\miA\mis
		\label{l0normopt}
	\end{opteq}	
 if $\miA\miR_{\bar{\NN}}\left(\mif^\star_{1:2K}\right)$ is injective  as a map  from $\C^{2K}\to \C^{N}$.
For the problem of identifying the  frequencies, $\mif^\star_{1:K}$,
associated to $\mis^\star$, several methods have been proposed in literature \cite{Stoica11,Hu12,Yang13,Tang13,Tan14,Yang16}. In the following we focus on the so-called gridless approaches \cite{Hu12,Yang16}. The gridless approach for $d=1$ is based on the Carathéodory-Fejér theorem. Such approach and the associated Carathéodory-Fejér theorem has been  recently generalized to the case of $d >1$ in \cite{Yang16} by introducing a \ac{PSD}  $d$-Level Toeplitz ($d$-LT) matrix, whose definition is provided in the following for the specific case of $d=3$.

\begin{definition}
\label{definition1}
Let $d=3$. A $\bar{N} \times \bar{N}$ matrix  $\miV \in \C^{\bar{N} \times \bar{N}}$ is a $d$-Level Toeplitz ($d$-LT) matrix
with nesting ordering  $\bar{\Zsf}\to\bar{\Ysf}\to\bar{\Xsf}$ if it is a 
$\bar{\Xsf} \times \bar{\Xsf}$ block Hermitian Toeplitz matrix defined as $\miV=\miV_{\bar{\XX}\bar{\YY}\bar{\ZZ}}$ in \eqref{eq:toeplitz}.(a). Furthermore, the $\bar{\Ysf} \bar{\Zsf} \times \bar{\Ysf} \bar{\Zsf}$-dimensional generic block $\miV_{\a\bar{\YY}\bar{\ZZ}}$ with  $-\bar{\Xsf}+1 \leq \a  \leq \bar{\Xsf}-1$ is a	$\bar{\Ysf} \times \bar{\Ysf}$ block Hermitian Toeplitz matrix of the form given in \eqref{eq:toeplitz}.(b) such that $\miV_{\a\bar{\YY}\bar{\ZZ}}= \miV_{-\a\bar{\YY}\bar{\ZZ}}^\dagger$. Finally, the generic block $\miV_{\a\b\bar{\ZZ}}$  with  $-\bar{\Xsf}+1 \leq \a  \leq \bar{\Xsf}-1$ and  $-\bar{\Ysf}+1 \leq \b  \leq \bar{\Ysf}-1$ is  a $\bar{\Zsf} \times \bar{\Zsf}$ Toeplitz matrix given in \eqref{eq:toeplitz}.(c) such that $\miV_{\a\b\bar{\ZZ}}= \miV_{-\a-\b\bar{\ZZ}}^\dagger$, and defined from  vector $\miv=\left[v_{\a\b (-\bar{\Zsf}+1)}, \dots,v_{\a\b(\bar{\Zsf}-1)}\right]\in\C^{(2\bar{\Zsf}-1)}$ with $v_{\a\b\c}=v_{-\a-\b-\c}^\dagger$ for $-\bar{\Xsf}+1 \leq \a  \leq \bar{\Xsf}-1$,  $-\bar{\Ysf}+1 \leq \b  \leq \bar{\Ysf}-1$ and $-\bar{\Zsf}+1 \leq \c  \leq \bar{\Zsf}-1$.

{\footnotesize	
\begin{equation}
\begin{array}{cccc}
\!\!\!\!\!\!
 \miV_{\bar{\XX}\bar{\YY}\bar{\ZZ}}=\begin{bmatrix}\miV_{0\bar{\YY}\bar{\ZZ}}&\dots &\miV_{(\bar{\Xsf}-1)\bar{\YY}\bar{\ZZ}}\\\miV_{(-1)\bar{\YY}\bar{\ZZ}}&\dots&\miV_{(\bar{\Xsf}-2)\bar{\YY}\bar{\ZZ}}\\\vdots & &\vdots \\ \miV_{(-\bar{\Xsf}+1)\bar{\YY}\bar{\ZZ}}&\dots & \miV_{0\bar{\YY}\bar{\ZZ}}\\ \end{bmatrix}
& \miV_{\a\bar{\YY}\bar{\ZZ}}=\begin{bmatrix}\miV_{\a 0\bar{\ZZ}} &\dots & \miV_{\a(\bar{\Ysf}-1)\bar{\ZZ}}\\\miV_{\a(-1)\bar{\ZZ}} &\dots&\miV_{\a(\bar{\Ysf}-2)\bar{\ZZ}}\\\vdots & &\vdots \\ \miV_{\a(-\bar{\Ysf}+1)\bar{\ZZ}}&\dots &\miV_{\a 0\bar{\ZZ}}\\ \end{bmatrix}
& \miV_{\a\b\bar{\ZZ}}=\begin{bmatrix}v_{\a\b 0} & \dots & v_{\a\b (\bar{\Zsf}-1)}\\v_{\a\b(-1)} & \dots&v_{\a\b (\bar{\Zsf}-2)}\\\vdots &  &\vdots \\ v_{\a\b (-\bar{\Zsf}+1)}& \dots & v_{\a\b 0}\\ \end{bmatrix} 
\vspace{0.2cm}\\
\quad \quad (a) & \quad \quad (b) &\quad \quad  (c)
\end{array}
	\label{eq:toeplitz}
	\end{equation}}

	For $d=2$, we define a $d$-LT matrix using \eqref{eq:toeplitz}.(b) and \eqref{eq:toeplitz}.(c) by fixing $\bar{\Xsf}=1$ so that $\a=0$ and $\miV=\miV_{0\bar{\YY}\bar{\ZZ}}$. Analogously, for $d=1$, we define a $d$-LT matrix from \eqref{eq:toeplitz}.(c) by  fixing $\bar{\Xsf}=1$ and $\bar{\Ysf}=1$ so that $\a=\b=0$  and $\miV=\miV_{00\bar{\ZZ}}$. 
	\QEDwhite
\end{definition}

We further introduce the  following definition: 
\begin{definition}
\label{definition2}
Given  a $\bar{N} \times \bar{N}$  $3$-LT   matrix $\miV$ 
with nesting ordering  $\bar{\Zsf}\to\bar{\Ysf}\to\bar{\Xsf}$, 
we say that $\miV$ has a canonical ordered structure (or equivalently 
$\miV$ is a canonical $d$-LT matrix) if  
there is a descending ordering in the component dimensions i.e.  $\bar{\Xsf}\leq\bar{\Ysf}\leq\bar{\Zsf}$.
\QEDwhite
\end{definition}


\begin{remark} 
\label{prettyremark}
There is an interesting connection, that we will later explore, 
between the Kronecker ordering of a steering vector  and the 
nesting ordering of a $3$-LT  matrix. Given a set of $K$ steering vectors $\mir_{\bar{\NN}}(\mif_k)$ with $k \in [K]$ 
with Kronecker ordering  $\bar{\Zsf}\to\bar{\Ysf}\to\bar{\Xsf}$ (see Remark \ref{Kronorder}), the ${\bar{N}} \times {\bar{N}}$ matrix defined as $\miV=\sum_{k=1}^K p_k \mir_{\bar{\NN}}(\mif_k)\mir_{\bar{\NN}}(\mif_k)^\dagger$ is a   \ac{PSD} 
$3$-LT  matrix with nesting ordering $\bar{\Zsf}\to\bar{\Ysf}\to\bar{\Xsf}$. 
\end{remark}
 
Using the previous definitions, we now overview some previous results. 
In \cite[Theorem~3]{Yang16} the authors show that given a $3$D-UD array  with $\NN=[\Xsf,\Ysf,\Zsf]$ elements, and its associated pair of sensing matrix $\miA_{\sf{U}}$ and virtual $3$D-UD array  of $\NNbar=[\bar{\Xsf},\bar{\Ysf},\bar{\Zsf}]$ elements, the recived signal $\mis^\star$ (the incoming signal $\miu^\star$) and its associated frequencies $\mif^\star_{1:K}$ can  be uniquely and perfectly reconstructed from the noiseless  measurement $\miy$ in \eqref{ec_parametrical_noiseless}, solving the following $\rank$ minimization  problem   \cite[Theorem~3, Remark~4]{Yang16}: 
\begin{opteq}
	\begin{aligned}
		\min_{r,  \mis \in \C^{\bar{N}}, \miS_{\bar{\XX}\bar{\YY}\bar{\ZZ}} }  \, \,&  \rank \left \{ \miS_{\bar{\XX}\bar{\YY}\bar{\ZZ}} \right\}  
		\quad\quad\quad{\rm s.t.} 	\quad 	&\begin{bmatrix}\miS_{\bar{\XX}\bar{\YY}\bar{\ZZ}}& \mis\\ \mis^\H & r\end{bmatrix}\succeq 0,		\quad & \miA_{\sf{U}}\mis= \miy,
	\end{aligned}
	\label{l0normopt_rank0}
\end{opteq}
\noindent 
if $K < \min\left\{\bar{\Xsf},\bar{\Ysf},\bar{\Zsf}\right\}$. In \eqref{l0normopt_rank0} $\miS_{\bar{\XX}\bar{\YY}\bar{\ZZ}}$ is a \ac{PSD}  $\dd$-LT matrix with arbitrary nesting order and $\miA_{\sf{U}} =  \big[\miI_N|{\bf 0}_{N\times (\bar{N}-N)} \big] {\bf{\Pi}_{\sf{U}}}$ with ${\bf{\Pi}_{\sf{U}}}$ a proper permutation matrix such that  only the first $\Xsf\leq\bar{\Xsf}$, $\Ysf\leq\bar{\Ysf}$, and $\Zsf\leq\bar{\Zsf}$ antennas 
are sensed. 
However, the results in  \cite[Theorem~3]{Yang16} only hold for 
$3$D-UD arrays
and rely  
(for the frequency recovery)  on the uniqueness of the Vandermonde decomposition of the resulting min-rank \ac{PSD} $\dd$-LT matrix, $\miS^\circ_{\bar{\XX}\bar{\YY}\bar{\ZZ}}$, which can only be guaranteed
if 
$\rank \left \{ \miS^\circ_{\bar{\XX}\bar{\YY}\bar{\ZZ}} \right\}< \min\left\{\bar{\Xsf},\bar{\Ysf},\bar{\Zsf}\right\}$ (see \cite[Theorem~1]{Yang16}).  


In the next section, we derive recovery results for 
arbitrary $\dd$D-\ac{AD} arrays and, exploiting the structure of our problem, we are able to reformulate \cite[Theorem~1]{Yang16} and \cite[Theorem~3]{Yang16} and  effectively enlarge the frequency recovery region, under much less restrictive conditions on the $\rank$ of a \ac{PSD}  $3$-LT matrix and on the dimensions of the virtual array. 


\subsubsection{{ Main results on recovery conditions}} 
\label{main_results} 


In this section, we provide the exact characterization of the frequency recovery region for arbitrary $\dd$D-\ac{AD} array, with $d\leq 3$ by identifying the conditions on the system parameters that guarantee, from the measurement $\miy$ in \eqref{ec_parametrical_noiseless}, perfect and unique recovery of the $K$--scatter vector $\mis^\star$ and its associated frequencies $\mif^\star_{1:K}\in \T^{d\times K}$.

Before stating our main results, 
let us present the following lemma, which will be used to prove our frequency recovery conditions and which also extends the result provided in \cite[Theorem~1]{Yang16}, deriving a less restrictive sufficient condition for the uniqueness of the Vandermonde decomposition of a $d$-level Toeplitz matrix.

\begin{lemma}
	\label{PrettyLemma}

Let  $\miS \in \C^{\bar{N} \times \bar{N}}$ be an $\bar{N} \times \bar{N}$  \ac{PSD}  $\dd$-LT matrix with $\rank$ $r <\max\left\{\bar{\Xsf},\bar{\Ysf},\bar{\Zsf}\right\}$, and with ordered canonical structure (as per Definition \ref{definition2}).
Denoting $W=\max\left\{\bar{\Xsf},\bar{\Ysf},\bar{\Zsf}\right\}$, if the $\rank$ of the  $W \times W$ upper-left corner\footnote{The $W \times W$ upper-left corner of $\miS$, is the $W \times W$ sub block of $\miS$ obtained  considering the first $W$ rows and the first $W$ columns of $\miS$.} of $\miS$ also equals $r$ then $\miS$ can be uniquely decomposed, via Algorithm \ref{Alg:PrettyLemma}, as $\miS=\miR_{\bar{\NN}}\left(\mif_{1:r}\right)\miP \miR^\dagger_{\bar{\NN}}\left(\mif_{1:r}\right),
$ with $\mif_{1:r}=\left\{\mif_{1},\mif_{2}\dots,\mif_{{ r}}\right\}\in\T^{d\times r}$ being a unique  set of frequencies, $\miR_{\bar{\NN}}\left(\mif_{1:r}\right)=\left[ \mir_{\bar{\NN}}(\mif_1), \mir_{\bar{\NN}}(\mif_2), \dots,  \mir_{\bar{\NN}}(\mif_r)\right]\in\C^{\bar{N} \times r}$ being the steering vector matrix associated to the vector frequencies $\mif_{1: r}$, and $\miP=\diag\left([p_1,\dots,p_r]\right)$, $p_k\in \R^+$ with $k\in[r]$.
\end{lemma}
\begin{proof}  
	The proof is given in Appendix \ref{appPrettyLemma}, where in a constructive way we prove existence and consequently uniqueness of the Vandermonde decomposition of $\miS$. Following the proof we are also able to develop an algorithm (Alg. \ref{Alg:PrettyLemma}) that can be applied now with less stringent conditions on the $\rank$ than in \cite{Yang16}.
\end{proof} 

\begin{algorithm}[t]
\caption{Algorithm for Lemma \ref{PrettyLemma}}
{\scriptsize
 \begin{algorithmic}[1]\label{Alg:PrettyLemma}
 \STATE \textbf{Input:}  $\miS=\miS_{\bar{\XX}\bar{\YY}\bar{\ZZ}}$ with $\rank\{\miS_{\bar{\XX}\bar{\YY}\bar{\ZZ}}\}=r$.\\
 \textbf{Step~1: Decomposing in the $\bar{\Xsf}$ dimension}
      \STATE Obtain the Cholesky decomposition of $\miS_{\bar{\XX}\bar{\YY}\bar{\ZZ}}=\miC_{\bar{\Xsf}\bar{\Ysf}\bar{\Zsf}}\miC_{\bar{\Xsf}\bar{\Ysf}\bar{\Zsf}}^\dagger$.
      \STATE Split $\miC_{\bar{\Xsf}\bar{\Ysf}\bar{\Zsf}}$ matrix row-wise in $\bar{\Xsf}$ blocks as $\miC_{\bar{\Xsf}\bar{\Ysf}\bar{\Zsf}}=[\miC_{0\bar{\Ysf}\bar{\Zsf}}^\dagger,\miC_{1\bar{\Ysf}\bar{\Zsf}}^\dagger,\dots,\miC_{(\Xsf-1)\bar{\Ysf}\bar{\Zsf}}^\dagger]^\dagger$.\\
      \IF{$\bar{\Xsf}>1$}
      \STATE Find the $\miU_{\bar{\Xsf}}$ unitary matrix such that $[\miC_{0\bar{\Ysf}\bar{\Zsf}}^\dagger,\miC_{1\bar{\Ysf}\bar{\Zsf}}^\dagger,\dots,\miC_{(\Xsf-2)\bar{\Ysf}\bar{\Zsf}}^\dagger]^\dagger\miU_{\bar{\Xsf}}=[\miC_{1\bar{\Ysf}\bar{\Zsf}}^\dagger,\miC_{2\bar{\Ysf}\bar{\Zsf}}^\dagger,\dots,\miC_{(\Xsf-1)\bar{\Ysf}\bar{\Zsf}}^\dagger]^\dagger$.
      \STATE Obtain the eigen-decomposition $\miU_{\bar{\Xsf}}=\miK_{\bar{\Xsf}}\miX\miK_{\bar{\Xsf}}^\dagger$ and we have that $\miX=\diag(e^{\j2\pi f^x_1},\dots,e^{\j2\pi f^x_r})$
      \ELSE
      \STATE Set $f^x_k=0$ $\forall k=[1,\dots,r]$
      \ENDIF\\
      \textbf{Step~2: Decomposing in the $\bar{\Ysf}$ dimension}
     
      \STATE Split $\miC_{0\bar{\Ysf}\bar{\Zsf}}$ matrix row-wise in $\bar{\Ysf}$ blocks as $\miC_{0\bar{\Ysf}\bar{\Zsf}}=[\miC_{00\bar{\Zsf}}^\dagger,\miC_{01\bar{\Zsf}}^\dagger,\dots,\miC_{0(\Ysf-1)\bar{\Zsf}}^\dagger]^\dagger$. 
      
      \IF{$\bar{\Ysf}>1$}
      \STATE Find the $\miU_{\bar{\Ysf}}$ unitary matrix such that $[\miC_{00\bar{\Zsf}}^\dagger,\miC_{01\bar{\Zsf}}^\dagger,\dots,\miC_{0(\Ysf-2)\bar{\Zsf}}^\dagger]^\dagger\miU_{\bar{\Ysf}}=[\miC_{01\bar{\Zsf}}^\dagger,\miC_{01\bar{\Zsf}}^\dagger,\dots,\miC_{0(\Ysf-1)\bar{\Zsf}}^\dagger]^\dagger$.
      \STATE Obtain the eigen-decomposition $\miU_{\bar{\Ysf}}=\miK_{\bar{\Ysf}}\miY\miK_{\bar{\Ysf}}^\dagger$ and we have that $\miY=\diag(e^{\j2\pi f^y_1},\dots,e^{\j2\pi f^y_r})$\\
      \ELSE
      \STATE Set $f^y_k=0$ $\forall k=[1,\dots,r]$
      \ENDIF\\
      \textbf{Step~3: Decomposing in the $\bar{\Zsf}$ dimension}
            \STATE Split $\miC_{00\bar{\Zsf}}$ matrix row-wise  and find the $\miU_{\bar{\Zsf}}$ unitary matrix such that $\miC_{00[0:\bar{\Zsf}-2]}\miU_{\bar{\Zsf}}=\miC_{00[1:\bar{\Zsf}-1]}$.
      \STATE Obtain the eigen-decomposition $\miU_{\bar{\Zsf}}=\miK_{\bar{\Zsf}}\miZ\miK_{\bar{\Zsf}}^\dagger$ and we have that $\miZ=\diag(e^{\j2\pi f^z_1},\dots,e^{\j2\pi f^z_r})$\\
      
      \textbf{Step~4: Obtain the set of frequencies $\mif_{1:r}$}
      \STATE Do the frequency pairing according to \cite{Yang16}.\\
      \STATE \textbf{Output:} The recovered frequencies $\mif_{1:r}=[f^{x\top}_{1:r} f^{y\top}_{1:r} f^{z^\top}_{1:r}]^\top$, the full $\rank$ matrix $\miR_{\bar{\NN}}\left(\mif_{1:r}\right)=\left[ \mir_{\bar{\NN}}(\mif_1), \dots,  \mir_{\bar{\NN}}(\mif_r)\right]$ and $\miP=\miR_{\bar{\NN}}^\g\left(\mif_{1:r}\right)\miS_{\bar{\XX}\bar{\YY}\bar{\ZZ}}\miR_{\bar{\NN}}^{\dagger\g}\left(\mif_{1:r}\right)$.
  \end{algorithmic}}
  \label{Alg1}
\end{algorithm}

Note that while Lemma \ref{PrettyLemma} is stated for $d \in \{1, 2,3\}$, it is worth to underline that the result holds even for $d>3$. 
\begin{remark}
Lemma \ref{PrettyLemma} admits a more general formulation where denoting by ${\miS} \in \C^{\bar{N} \times \bar{N}}$ an $\bar{N} \times \bar{N}$  \ac{PSD}  $\dd$-LT matrix with rank $r$, and denoting by $W$ the dimension of the largest $1$-LT upper-left corner of $\miS$. If the  $\rank$ of such block also equals $r$ and $r <W$, then  $\miS$ admits a unique Vandermonde decomposition.
\end{remark}

\begin{remark}
Lemma \ref{PrettyLemma} shows that while the sufficient condition for the uniqueness of the Vandermonde decomposition ($r<\min\left\{\bar{\Xsf},\bar{\Ysf},\bar{\Zsf}\right\}$) in \cite[Theorem~1]{Yang16} is tight for $d=1$, this would not hold any more for a \ac{PSD}  $\dd$-LT matrix $\miS \in \C^{\bar{N} \times \bar{N}}$ with $d\geq 2$. In fact if $\miS \in \C^{\bar{N} \times \bar{N}}$ follows Lemma \ref{PrettyLemma} conditions, then a sufficient condition  for $\miS$ to admit a unique Vandermonde decomposition is the less restrictive condition $r< \max\left\{\bar{\Xsf},\bar{\Ysf},\bar{\Zsf}\right\}$, compared to $r< \min\left\{\bar{\Xsf},\bar{\Ysf},\bar{\Zsf}\right\}$ \cite[Theorem~1]{Yang16}. This interesting result essentially stems from the fact that the proposed decomposition approach relies on the  decomposing the upper-left block matrix in $\miS$, cf. lines~17 and 18 of Algorithm~\ref{Alg:PrettyLemma}. Therefore, 
the canonical ordering required in Lemma \ref{PrettyLemma} poses the base to allow the less restrictive condition ($r< \max\left\{\bar{\Xsf},\bar{\Ysf},\bar{\Zsf}\right\}$) on the uniqueness of the decomposition and hence enlarges the set of \ac{PSD}  $\dd$-LT matrix for which we can guarantee unique decomposition. However, if the ordering does not follow Definition \ref{definition2}, as required in 
Lemma \ref{PrettyLemma}, one has to enforce
the condition stated in \cite[Theorem~1]{Yang16}, i.e. $r< \min\left\{\bar{\Xsf},\bar{\Ysf},\bar{\Zsf}\right\}$, to ensure the decomposability of the upper-left block matrix.  To the best of the authors' knowledge, the result stated in Lemma \ref{PrettyLemma} is the tightest  condition for the uniqueness decomposition of \ac{PSD}  $\dd$-LT matrices in the literature.
\end{remark}



In addition to Lemma \ref{PrettyLemma}, in order to state our recovery conditions, we also need to provide some few useful definitions and assumptions.

Specifically, in the following, we always consider as Kronecker ordering for the virtual $3$D-UD streering vectors $\bar{\Zsf}\to\bar{\Ysf}\to\bar{\Xsf}$ and we always assume without loss of generality $\bar{\Xsf}\leq \bar{\Ysf} \leq \bar{\Zsf}$. Note that in case this assumption would not initially hold, we can always operate a proper rotation of the Cartesian system (which corresponds to a permutation of the elements of the measurement vector $\miy$, or a different reading of the antenna elements of the $\dd$-AD). This will ensure that we can always enforce the canonical ordering of Definition \ref{definition2}  on the set of \ac{PSD} matrices, in the subsequent proposed $\min$-$\rank$ optimization.

Furthermore, being $\mathcal{A}$ the sensing set defined in \eqref{eq:sensing_set}, let us introduce the following definition:
\begin{definition}
\label{setinjective}
$ \mathcal{A}^{(K)} \subset \mathcal{A}$ denotes  the subset of all binary $N \times \bar{N}$ sensing matrices $\mathcal{A}$ for which  it exists a subset of rows of $\miA\in\mathcal{A}$, say $ {\mathcal{L}\subseteq[N]}$ with $ 2K <  |{\mathcal{L}}| \leq N$ such that: 
\begin{itemize}
\item[\bf 1.C] under Assumption \ref{A1}, 
$\miA^{{(\mathcal{L})}} \miR_{\bar{\NN}}(\mif_{1:2K})$ is injective as map  from $\C^{2K}\to \C^{|{\mathcal{L}}|}$;
\item[\bf 2.C]  denoted by $\mif^\star_{1:K}$ a set of $K$ frequencies satisfying Assumption \ref{A1} and by  $\mif^\diamond_{1:r^\diamond}$ an arbitrary  set of ${ r^\diamond}  \leq K$  frequencies such that for any  $k \in [K]$ and $j\in [{ r^\diamond}]$,  $ \mif^\diamond_{j}\neq \mif^\star_{k}$, $\rank\left\{\miA^{{(\mathcal{L})}} \miR_{\bar{\NN}}({\mif^{\star\diamond}_{1:K+r^\diamond}} ) \right\}= K +\rank\left\{\miA^{{(\mathcal{L})}} \miR_{\bar{\NN}}({\mif^\diamond_{1:{ r^\diamond}})}\right\}$,
with ${\mif^{\star\diamond}_{1:K+ r^\diamond}=[\mif^\star_{1:K} \,\mif^\diamond_{1:r^\diamond}]}$.{\QEDwhite}
\end{itemize}
\end{definition}

In the following we provide, through Lemma~\ref{Lemma:theoremApiProbalistic} and 
Lemma~\ref{Lemma:Corollary} (see Appendices \ref{proofcorollary} and \ref{appAIW0probabilistic}), 
a class of sensing matrices of practical interest in $\mathcal{A}^{(K)}$. Specifically, in Lemma~\ref{Lemma:theoremApiProbalistic} (see Appendix \ref{appAIW0probabilistic}), we show that, provided that the number of  elements in the $\dd$D-\ac{AD} array is sufficiently large, $\mathcal{A}^{(K)}\equiv \mathcal{A}$ with high probability. In Lemma~\ref{Lemma:Corollary} we identify a second class of matrices,  say $\mathcal{A}^{(K)}_{\Csf}\subset\mathcal{A}$,
for  which we prove  that $\mathcal{A}^{(K)}_{\Csf} \subset\mathcal{A}^{(K)}$.  Such class is identified by the so-called 
\well~sensing matrices as in the next two definitions:


\begin{definition}
\label{def:wellstruct}
A permutation matrix  ${\bf \Pi}$ belongs to the \well~permutation class $ \mathcal{P}_{\Csf}$,  and note it as ${\bf \Pi}_\Csf$ if:
\begin{itemize}
\item[\emph{i)}]  denoting by $\mathcal{I}$ the set of indices for the columns of  ${\bf \Pi}_{\Csf} $  that have a $1$ among its first $N$ rows\footnote{ Note that $\mathcal{I}$ identifies the components of the steering vector $\mir_{\bar{\NN}}(\mif)$ that are going to be sensed.}, i.e., $j \in  \mathcal{I}$ if $\exists\, n \in [N]$ such that $ {\pi}_{nj}=1$,
\item[\emph{ii)}]  denoting by $\mathcal{I}_{\Csf} \subseteq
\mathcal{I}$ the  largest subset of ordered indices in $\mathcal{I}$  not necessarily consecutive, such that for any arbitrary vector $\mif=[f^x, f^y, f^z ]^\top \in \T^{d}$, the $N_{\Csf}$ dimensional vector containing the elements of $\mir_{\bar{\NN}}(\mif)$ whose indices belong to ${ \mathcal{I}_{\Csf}}$ can be written as
\begin{equation}
\mir^{{(\mathcal{I}_{\Csf})}}_{\bar{\NN}}(\mif)= 
e^{\j2\pi {(\m_x f^{x}+ \m_y f^{y}+ \m_z f^{z})}} 
\mir_{{\NN}_{\Csf}}(\mif^{\Csf})
\label{eq_ws}
\end{equation}
where $\mir_{{\NN}_{\Csf}}(\mif^{\Csf})$ corresponds to a $\dd$D-UD with structure given by $\NN_\Csf=[\Xsf_\Csf,\Ysf_\Csf,\Zsf_\Csf]$,
and $\mif^{\Csf}=\left[\ell_x f^x,\ell_y f^y,\ell_z f^z\right]^\top
$ is  a $d$-dimensional vector with $\ell_\alpha$ a proper positive integer, $\alpha \in (x,y,z)$,  ${\m_x} \in [{\bar{\Xsf}}]$,
${\m_y} \in [{\bar{\Ysf}}]$, ${\m_z} \in [{\bar{\Zsf}}]$
and finally $S_{\Csf}=\Xsf_\Csf+\Ysf_\Csf+\Zsf_\Csf$, $N_{\Csf}=|{ \mathcal{I}_{\Csf}}|=\Xsf_\Csf\Ysf_\Csf\Zsf_\Csf$, 
\end{itemize}
we have that $S_\Csf \geq (d+1)$ or $N_{\Csf}\geq 2$.  
\QEDwhite
\end{definition}
From Definition \ref{def:wellstruct}, it follows that:  
\begin{definition}
\label{def:wellstruct1}
The set of \well~sensing matrices is defined as $\mathcal{A}_\Csf= \big\{\miA\in\{0,1\}^{N\times \bar{N}}: \, \big[\miI_N|{\bf 0}_{N\times (\bar{N}-N)}\big]  {\bf \Pi}_\Csf,\, {\bf \Pi}_\Csf \in\mathcal{P}_\Csf\big\},$
with $\mathcal{P}_\Csf$ as in Definition
\ref{def:wellstruct}.
\end{definition}

\begin{example}
Let's assume the planar array from \FIG{fig:array}.a, whose encompassing virtual array is given by the $2$D-UD structure characterized by $\bar{\Xsf}=1$ $\bar{\Ysf}=3$ and $\bar{\Zsf}=4$, and where the three antennas in red are removed by means of the sensing matrix. In this configuration we have that $\bar{N}=12$, $N=9$ and we can find $\mathcal{I}=\{2:5,8:12\}$, $\mathcal{I}_\Csf=\{2:4,10:12\}$, $\NN_{\Csf}=[1,2,3]$, $\mif^{\Csf}=\left[f^x,2 f^y,f^z\right]^\top$ and $\m_x=\m_y=0$, $\m_z=1$,  such that $S_\Csf=6$ and $N_\Csf=6$. Then we say that the array in \FIG{fig:array}.a, is a \well~array.
\end{example}

\begin{example}
\label{ex2}
A special subset  of $\mathcal{P}_{\Csf}$  is the set of all permutations matrices, say ${\bf \Pi}_{\sf U}$, whose associated sensing matrix $\miA_{\sf{U}} =  \big[\miI_N|{\bf 0}_{N\times (\bar{N}-N)} \big] {\bf{\Pi}_{\sf{U}}}$ is the one where only the first $\Xsf\leq\bar{\Xsf}$, $\Ysf\leq\bar{\Ysf}$, and $\Zsf\leq\bar{\Zsf}$ antennas are sensed.
\end{example}

We are now ready to state our main results:

\begin{theorem}
	\label{theoremApi}
Consider a linear measurement model as in \eqref{ec_parametrical_noiseless}.
Under Assumption \ref{A1}, provided that $ K<  \max\{\bar{\Xsf}, \bar{\Ysf}, \bar{\Zsf}\}$,  if  $\miA\in \mathcal{A}^{(K)}$ as in Definition \ref{setinjective}, the $K$--scatter vector $\mis^\star$ and its associated 
 frequencies  $\mif^\star_{1:K}\in \T^{d\times K}$  
can be uniquely identified by solving the optimization problem:
	\begin{opteq}
	\begin{aligned}
		\min_{r,  \mis \in \C^{\bar{N}}, \miS_{\bar{\XX}\bar{\YY}\bar{\ZZ}} \in {\mathcal{T}_{\bar{\Xsf}\bar{\Ysf}\bar{\Zsf}}} }  \, \,&  \rank \left \{ \miS_{\bar{\XX}\bar{\YY}\bar{\ZZ}} \right\}\quad \quad   
		{\rm s.t.} 	\quad 	&\begin{bmatrix}\miS_{\bar{\XX}\bar{\YY}\bar{\ZZ}}& \mis\\ \mis^\H & r\end{bmatrix}\succeq 0,\quad & \miA \mis= \miy.
	\end{aligned}
	\label{l0normopt_rank}
	\end{opteq}
where  $\mathcal{T}_{\bar{\Xsf}\bar{\Ysf}\bar{\Zsf}} \subseteq \C^{\bar{N} \times \bar{N}} $ denotes the set of all canonical \ac{PSD}  $d$-LT matrices  (see Definition \ref{definition2}). The frequencies  $\mif^\star_{1:K}\in \T^{d\times K}$ can be  obtained by Vandermonde Decomposition, via Algorithm \ref{Alg:PrettyLemma}, of  the $\dd$-LT matrix $\miS_{\bar{\XX}\bar{\YY}\bar{\ZZ}}$ solution of \eqref{l0normopt_rank}.
\end{theorem}

\begin{proof}  
The proof  is provided in Appendix \ref{appAIW0}.
\end{proof} 

\begin{remark}
\label{vafanculo}
Compared to the result provided in  \cite[Theorem~3]{Yang16}, Theorem \ref{theoremApi} is able to extend the region of resolvable scatters. In fact, differently from \cite{Yang16}, where a sufficient condition for recovery is provided consisting of 
$K< \min\{\bar{\Xsf}, \bar{\Ysf}, \bar{\Zsf}\}$, we significantly enlarge this condition to
$K<  \max\{\bar{\Xsf}, \bar{\Ysf}, \bar{\Zsf}\}$. To do this, we first \emph{i)} restructure the $N$-dimensional observable vector as the sampled version of a linear combination of steering vectors with Kronecker ordering $\bar{\Zsf}\to\bar{\Ysf}\to\bar{\Xsf}$
 such that $\bar{\Zsf}\geq\bar{\Ysf}\geq\bar{\Xsf}$, \emph{ii)} we search for the $\min$-$\rank$ matrix  $\miS^\circ_{\bar{\XX}\bar{\YY}\bar{\ZZ}}$, solution of \eqref{l0normopt_rank}, among all canonical  $\dd$-LT matrices and finally \emph{iii)} we decompose the canonical matrix $\miS^\circ_{\bar{\XX}\bar{\YY}\bar{\ZZ}}$, via Algorithm \ref{Alg:PrettyLemma} to extract relevant frequency parameters. This procedure ensures a significant enlargement in the number of scatters $K$ that can be resolved, together with the error-free identification of the associated frequencies  $\mif^\star_{1:K}\in \T^{d\times K}$  providing information on the multi-dimensional propagation parameters of the aforementioned scatters.
\end{remark}

\begin{corollary}
\label{CorollaryTh1}
Theorem \ref{theoremApi} holds  for 
the special class of \well~sensing  matrices (as per Definition \ref{def:wellstruct1}), such that $S_\Csf -(\dd-1)\geq 2K$. We denote  such class of sensing matrices by $\mathcal{A}^{(K)}_{\Csf}\subset\mathcal{A}_{\Csf}$.
\end{corollary}
\begin{proof}  
The proof  is provided in Appendix \ref{proofcorollary}.
\end{proof} 

Corollary \ref{CorollaryTh1} refers to
the class, $\mathcal{A}^{(K)}_{\Csf}\subset\mathcal{A}_{\Csf}$, of \well~sensing matrices for which $S_\Csf \geq 2K+(\dd-1)$, and 
states that if $\miA \in \mathcal{A}^{(K)}_{\Csf}$ 
the region of resolvable scatters
is $K< \min\{\frac{S_\Csf -(\dd-1)}{2}, \max\{\bar{\Xsf}, \bar{\Ysf}, \bar{\Zsf}\}\}$. An enlarged class of \well~sensing matrices which contains $\mathcal{A}^{(K)}_{\Csf}\subset\mathcal{A}_{\Csf}$   is the class of \well~sensing matrices such that $N_\Csf > 2K$.
For this enlarged class we conjecture that the region of resolvable scatters is enlarged to  $K<
\min\{\frac{N_\Csf}{2}, \max\{\bar{\Xsf}, \bar{\Ysf}, \bar{\Zsf}\}\}$, including the region stated in Corollary \ref{CorollaryTh1}.
More formally:
\begin{conjecture}
\label{conjec1}
Theorem \ref{theoremApi} holds for the subset of \well~sensing matrices (as per Definition \ref{def:wellstruct1}), such that  $N_\Csf>2K$.
\end{conjecture}

Conjecture \ref{conjec1} validity will be shown by simulation in Section \ref{sec:Vb}. 

The previous results consider a class of sensing matrices with a specific structure which in some practical applications could not be satisfied or even difficult to verify.  A natural question arises on how previous results extend to the case of sensing matrices with arbitrary measuring structures, if they do so. Next theorem answers to this question showing that frequency recovery is still possible, even in the presence of arbitrary sensing matrices, if the dimension of the measurement $N$ is larger than a certain threshold which linearly grows with the number of scatters we want to resolve. This implies that we can have a linear trade-off between hardware cost and resolvable capability.
This result is of significant practical  relevance specially in the context of intelligent smart surfaces with a massive distribution of antennas \cite{Hu18}. 

\begin{theorem}
	\label{theoremApiProbalistic}
 Given a linear measurement model as in \eqref{ec_parametrical_noiseless}  with $\miA \in  \mathcal{A}$ being an \textit{arbitrary} binary sensing matrix as defined in \eqref{eq:sensing_set}, under Assumption \ref{A1}, 
	 there exist with probability $1-\epsilon$ a unique $K$--scatter vector $\mis^\star$ as the solution to problem \eqref{l0normopt_rank}, and  a  unique frequency set  $\mif^\star_{1:K}\in \T^{d\times K}$ satisfying
	 \eqref{ec_parametrical_noiseless},  provided that, \emph{i)} $ K< \max\{\bar{\Xsf}, \bar{\Ysf}, \bar{\Zsf}\}$
	 and \emph{ii)} $N\geq 2KC\log\left(2K\epsilon^{-1}\right)$ with $C$ being a proper constant which is not larger than $12$.
	 The frequency set  $\mif^\star_{1:K}\in \T^{d\times K}$ is  obtained by Vandermonde Decomposition, via Algorithm \ref{Alg:PrettyLemma}, of  the $\dd$-LT matrix $\miS_{\bar{\XX}\bar{\YY}\bar{\ZZ}}$ solution of \eqref{l0normopt_rank}.
\end{theorem}

\begin{proof}  
The proof  is provided in Appendix \ref{appAIW0probabilistic}.
\end{proof} 

\subsubsection{Reconstruction via  Convex Relaxation} 
In the previous section, we focused on the atomic $\ell_0$-\ac{AN} reconstruction. However, the $\ell_0$-\ac{AN} minimization, as well as its equivalent reformulation in terms of  a $\rank$ minimization problem, are NP-hard \cite{CSbook}. Nonetheless, keeping in mind that the atomic $\ell_q$-norm  approaches the atomic $\ell_0$-norm as $q$ tends to zero, the usual approach is to approximate  \eqref{l0normopt},  or  equivalently  \eqref{l0normopt_rank}, via the following optimization problem: 

\begin{opteq}
	\underset{\mis \in \C^{\bar{N}}}{\min} \left\|\mis\right\|_{\mathcal{R},1}  \quad {\rm s.t.}  \quad  \miA\mis= \miy,
	\label{covex}
\end{opteq}
where $\|\mis\|_{\mathcal{R},1}$ denotes  the  atomic $\ell_1$-norm  ($\ell_1$-\ac{AN}) \cite{Chandrasekaran2012,Chi15}  of a vector $\mis$ in  $
\mathcal{R}=\left\{\mir_{\bar{\NN}}(\mif):\mif\in\T^d\right\}$ given by $\|\mis\|_{\mathcal{R},1}=\inf_{\mif_k\in \T^d,u_k\in\C}\bigg\{ \sum_k|u_k| :\, \,   \mis=\sum_k u_k\mir_{\bar{\NN}}\left(\mif_k\right)\bigg\}$.

In  \cite{Tang13},  a semidefinite characterization of the $\ell_1$ atomic norm was introduced for $d=1$ and generalized to arbitrary $d$ dimensions in \cite{Yang16} as follows:
\begin{opteq}
	\begin{aligned}
		\min_{t,  \mis \in \C^{\bar{N}}, \miT_{\bar{\XX}\bar{\YY}\bar{\ZZ}} \in \mathcal{T}_{\bar{\Xsf}\bar{\Ysf}\bar{\Zsf}}}& \, \,  \frac{1}{2}t+\frac{1}{2}\tr\left\{\miT_{\bar{\XX}\bar{\YY}\bar{\ZZ}} \right\} \quad \quad 
		{\rm s.t.} 	\quad 	&\begin{bmatrix}\miT_{\bar{\XX}\bar{\YY}\bar{\ZZ}} & \mis\\ \mis^\H & t\end{bmatrix}\succeq 0, \quad &\miA\mis= \miy.
	\end{aligned}
	\label{l1normopt_convex}
\end{opteq}

In the next two theorems, we provide recovery results for the  convex opimization problem \eqref{l1normopt_convex}, both for the sensing set $\mathcal{A}^{(K)}_\Csf$ and for an arbitrary sensing set where the dimension of the measurement $N$ is larger than a threshold. These scenarios would be analogous respectively to Corollary \ref{CorollaryTh1} and Theorem \ref{theoremApiProbalistic} for the non-convex problem \eqref{l0normopt_rank}.

\begin{theorem}
	\label{theoremL1}
	
	Consider a linear measurement model as in \eqref{ec_parametrical_noiseless}.  Under Assumption \ref{A1}, if $\miA\in \mathcal{A}^{(K)}_\Csf$ as defined in Corollary \ref{CorollaryTh1}, 
	there exists a unique $K$--scatter vector $\mis^\star$ as the solution to problem \eqref{l1normopt_convex} and a unique frequency set  $\mif^\star_{1:K}\in \T^{d\times K}$ satisfying $\miy=\miA \sum_k u^\star_k\mir_{\bar{\NN}}\left(\mif^\star_k\right)$ provided that, \emph{i)}
$ K<\max\{\bar{\Xsf}, \bar{\Ysf}, \bar{\Zsf}\}$ and $K\leq (S_\Csf-(\dd-1))/2$,  \emph{ii)} the optimal solution to \eqref{l1normopt_convex}, denoted by $\left( t^\circ,  \mis_{\ell_1}^\circ, \miT^\circ_{\bar{\XX}\bar{\YY}\bar{\ZZ}}\right)$,  also satisfies $r_{\ell_1}^\circ< \max\{\bar{\Xsf}, \bar{\Ysf}, \bar{\Zsf}\}$ and $r_{\ell_1}^\circ\leq (S_\Csf-(\dd-1))/2$ with 
$r_{\ell_1}^\circ  \triangleq\rank \left \{ \miT^\circ_{\bar{\XX}\bar{\YY}\bar{\ZZ}}\right\}$
\end{theorem}
\begin{proof}  
Theorem \ref{theoremL1} is easily proved following the same derivations of Theorem \ref{theoremApi} and making use of the $\rank$ condition of the solution in \eqref{l1normopt_convex}.
\end{proof}

Analogously to what considered for Corollary
\ref{CorollaryTh1}, we have numerically verified the following conjecture:

\begin{conjecture}
\label{conjec2}
Theorem  \ref{theoremL1} holds for the subset of \well~sensing matrices (as per Definition \ref{def:wellstruct1}), provided that
\emph{i)}
$ K<\min\{ \frac{N_\Csf}{2},\max\{\bar{\Xsf}, \bar{\Ysf}, \bar{\Zsf}\}\}$,
 and \emph{ii)} $r_{\ell_1}^\circ< \min\{ \frac{N_\Csf}{2} ,\max\{\bar{\Xsf}, \bar{\Ysf}, \bar{\Zsf}\}\}$.
\end{conjecture}

\begin{theorem}
\label{theoremApiProbalisticl1}

Given a linear measurement model as in \eqref{ec_parametrical_noiseless}  with $\miA \in  \mathcal{A}$ being an \textit{arbitrary} binary sensing matrix as defined in \eqref{eq:sensing_set}, under Assumption \ref{A1}, 
	 there exist with probability $1-\epsilon$ a unique $K$--scatter vector $\mis^\star$ as the solution to problem \eqref{l1normopt_convex}, and  a  unique frequency set $\mif^\star_{1:K}\in \T^{d\times K}$ satisfying
	 \eqref{ec_parametrical_noiseless},  provided that,
	 \emph{i)} $ K< \min\{\frac{N}{2}, \max\{\bar{\Xsf}, \bar{\Ysf}, \bar{\Zsf}\}\}$, \emph{ii)} $N\geq 2KC\log\left(2K\epsilon^{-1}\right)$ with $C$ being a proper constant which is not larger than $12$, and \emph{iii)} the optimal solution to \eqref{l1normopt_convex} noted as $\left( t^\circ,  \mis_{\ell_1}^\circ, \miT^\circ_{\bar{\XX}\bar{\YY}\bar{\ZZ}}\right)$ with $\rank \left \{ \miT^\circ_{\bar{\XX}\bar{\YY}\bar{\ZZ}}\right\}\triangleq r_{\ell_1}^\circ$ also satisfies  $r_{\ell_1}^\circ<\min\{\frac{N}{2}, \max\{\bar{\Xsf}, \bar{\Ysf}, \bar{\Zsf}\}\}$.
\end{theorem}
\begin{proof}  
	The proof is similar to the proof of Theorem \ref{theoremApiProbalistic} except that following Lemma~\ref{Lemma:theoremApiProbalistic}, we exploit that $\miA\in\mathcal{A}^{(K)}$ holds with probability $1-\epsilon$ under condition \emph{ii)} $N\geq 2KC\log\left(2K\epsilon^{-1}\right)$.
\end{proof}

\subsection{Improved convex relaxation via $\ell_2+\ell_1$-\ac{AN} optimization}

In this section, we highlight some shortcomings of the convex recovery approach in \eqref{l1normopt_convex} and formulate an alternative convex optimization to address these issues. 
\emph{i)} Recall that the recovery problem \eqref{l1normopt_convex} is formulated under a noise-free assumption.  \emph{ii)} The conditions in  Th. \ref{theoremL1}  guarantee an error--free recovery under a $\rank$ condition on the solution to the optimization problem which ensures a unique Vandermonde decomposition of $\miT^\circ_{\bar{\XX}\bar{\YY}\bar{\ZZ}}$ that assures the  sparse nature of the minimizer 
$\mis_{\ell_1}^\circ$ solution to
\eqref{l1normopt_convex}\footnote{Recall, in fact,  that $\ell_1$-atomic norm minimizers are not always sparse.} .  

To address the aforementioned challenges, we formulate a convex recovery problem as the extension of \eqref{l1normopt_convex} by formulating a convex combination, parameterized in a regularization parameter $\tau$, of the the distance between $\miA\mis$ and $\miy$ (i.e $\ell_2$-norm of the difference  $\miy -\miA\mis$) and the $\ell_1$-norm of $\mis$, i.e.,
\begin{opteq}
	\begin{aligned}
		\min_{t,  \mis \in \C^{\bar{N}}, \miT_{\bar{\XX}\bar{\YY}\bar{\ZZ}} \in {\mathcal{T}_{\bar{\Xsf}\bar{\Ysf}\bar{\Zsf}}} }  \!\! &(1-\tau)\|  \miA \mis -  \miy \|_2^2+\tau \left(\frac{1}{2}t+\frac{1}{2}\tr\left\{\miT_{\bar{\XX}\bar{\YY}\bar{\ZZ}}\right\}\right)  \quad 
		{\rm s.t.} 		&\begin{bmatrix}\miT_{\bar{\XX}\bar{\YY}\bar{\ZZ}} & \mis\\ \mis^\H & t\end{bmatrix}\succeq 0.
	\end{aligned}
	\label{eq:SDPnoisy}
\end{opteq}
where $\tau$ in \eqref{eq:SDPnoisy}, can be properly optimized to minimize the average frequency recovery error. Let $\tau_o$ be such optimal value. Clearly $\tau_o$ is function of the $\snr$ defined as 
$\snr= \frac{E[\|\miu^*\|^2]}{\sigma^2}$.
Unfortunately, an explicit expression for $\tau_o$ is hard to derive. However,   upper and lower bounds of the optimal  $\tau$  have been provided in \cite{Bhaskar13} as a function of $\snr$:
{\small$$
\hspace{-1cm}\tau_u=\frac{2\sigma\left(1+\frac{1}{\log(N)}\right)\sqrt{N\log(N)+N\log(4\pi\log(N))}}{1+2\sigma\left(1+\frac{1}{\log(N)}\right)\sqrt{N\log(N)+N\log(4\pi\log(N))}},\quad
\tau_l=\frac{2\sigma\sqrt{N\log(N)-\frac{N}{2}\log(4\pi\log(N))}}{1+2\sigma\sqrt{N\log(N)-\frac{N}{2}\log(4\pi\log(N))}}
$$
}
with $\tau_l\leq\tau_o\leq\tau_u$.

Using classical tools from noise sparse representation, bounds on the frequency recovery errors  can be derived  \cite{CSbook}. Due to space constraints, we do not include these generalizations in this paper and leave this topic as a potential direction for future research. Instead, in Section~\ref{sec:V}, we investigate the performance of the convex recovery problem in \eqref{eq:SDPnoisy} via simulation. 

\section{ 
Application scenarios and frequency recovery performance
}
\label{sec:V}

The definition of \well~sensing matrices (i.e., Definition \ref{def:wellstruct}) provides an operational characterization of Definition 
\ref{setinjective} in the sense that Definition \ref{def:wellstruct}
identifies simple parameters that allow establishing rules that facilitate verifying the conditions needed for signal recovery uniqueness (conditions 1.C and 2.C in Section \ref{subsect:mavafanculo}). 
The way in which the definition of \well~sensing matrices translates into the physical structure of an arbitrary non-uniform $3$D array (referred from now on as \emph{physical array}) is in finding an underlying uniform antenna structure, embedded in the non-uniform array (see examples in sections \ref{main_results} and \ref{sec:Vb}) and its associated geometric parameters whose relationship with the number of scatters guarantee the possibility of signal recovery. More precisely, given an arbitrary non-uniform physical array, in order to identify its resolvable region, i.e., the number of scatters that can be resolved,
the first step is to find a uniform sub-array which corresponds to the largest uniform structure embedded in the original non-uniform deployment, characterized by $\NN_\Csf=[\Xsf_\Csf,\Ysf_\Csf,\Zsf_\Csf]$, $S_\Csf=\Xsf_\Csf+\Ysf_\Csf+\Zsf_\Csf$ and $N_\Csf=\Xsf_\Csf\Ysf_\Csf\Zsf_\Csf$, where $\Xsf_\Csf$,$\Ysf_\Csf$,$\Zsf_\Csf$ represent the number of antennas in each dimension. Based on these physical parameters, we prove (see  Corollary \ref{CorollaryTh1} and Theorem \ref{theoremL1}) that any number of scatters smaller or equal than $(S_\Csf - (d-1))/2$  can be resolved.  We strongly conjecture, as validated by our numerical results (see Section \ref{sec:Vb}),
that the resolvable region is in fact $N_\Csf/2$ (see Conjectures \ref{conjec1}-\ref{conjec2}).
While evaluating the resolvable region requires identifying an embedded uniform sub-array, in order to recover the multi-dimensional propagation parameters, we also need to identify an encompassing virtual uniform array, i.e., an expanded virtual uniform structure that contains the original physical structure as a sub-array, characterized by $\bar{\NN}=[\bar{\Xsf},\bar{\Ysf},\bar{\Zsf}]$ and $\bar{N}=\bar{\Xsf}\bar{\Ysf}\bar{\Zsf}$ with $\bar{\Xsf}$,$\bar{\Ysf}$,$\bar{\Zsf}$ denoting the number of antennas in each dimension, with $\bar{\Zsf}$ the dominant dimension by convention. 
A proper identification of the virtual uniform array allows recovering the propagation parameters of a given number of possible scatters up to the physical limit  $K \leq (S_\Csf -(d-1))/2$ (or $K <N/2-1$ as we conjecture)
if $\bar{\Zsf}> K$.
Since the  dimensions of the virtual array influence the computational complexity of the proposed recovery algorithms,
they can be seen as tunable parameters that enable trading computational complexity with frequency-recovery performance.
Furthermore, as already pointed in Section \ref{Sec:Contributions}, 
our results provide guidelines for the design of the non-uniform array geometry. In fact, setting a geometric deployment of
the array such that $\Zsf_\Csf= 2K$, $\Ysf_\Csf=2$ and $\Xsf_\Csf\leq 2$, and given that the complexity of \eqref{l1normopt_convex} and \eqref{eq:SDPnoisy}  scales 
as $\gamma=3.5$  \cite{Krishnan2005} with the dimension of the problem variables, we have that up to  $K =(S_\Csf - (d-1))/2$ scatters  can be resolved with  $O(K^{\gamma})$ complexity, significantly improving existing results that require $O(K^{\dd\gamma})$ complexity. 

{
}

\subsection{Resolvable region and recovery performance}
\label{sec:Vb}
We provide next some numerical results to show the resolvable region and the frequency recovery performance of $\ell_0$-AN and $\ell_1$-AN for different $2$D and $3$D physical antenna deployments. Given a set of measurements following \eqref{ec_parametrical}, the error $\frac{1}{d}\|\mif^\circ_{1:K}-\mif^\star_{1:K}\|_{1}$ between the recovered $d$-dimensional frequencies $\mif^\circ_{1:K}$ and the real frequencies $\mif^\star_{1:K}$  
is computed and presented both for noiseless and noisy scenarios. In all the results following, the number of scatters $K$ is assumed to be known. 

\begin{table}[t]
\begin{center}
{\scriptsize
\begin{tabular}{l l c c}
\toprule
\quad \textbf{Physical array} & \quad \quad\textbf{Physical parameters} & Corollary \ref{CorollaryTh1}  & Conjecture \ref{conjec1}  \\
 &  & $K\leq \bigg\lfloor
\frac{S_\Csf-(\dd-1)}{2}\bigg\rfloor$ & $K\leq\bigg\lceil \frac{N_\Csf}{2}-1\bigg\rceil$ \\
\midrule
$2$D-UD& $\NN=\NN_\Csf=\left[1,3,6\right]$. $S_\Csf=10$, $N_\Csf=18$ &  $K\leq 4$ & $K\leq 8$ \\
\midrule
$3$D-AD cubic& $\NN_\Csf=\left[2,4,4\right]$, $S_\Csf=10$, $N_\Csf=32$ &  $K\leq 4$ & $K\leq 15$ \\
\bottomrule
\end{tabular}}
\caption{Resolvable region under Corollary \ref{CorollaryTh1} and Conjecture \ref{conjec1} for the physical array configurations explored.} 
\label{resolvable}
\end{center}
\vspace{-1.5cm}
\end{table}

\subsubsection{Noiseless measurements}
\label{Sec:VNoiseless}

\begin{figure}[htbp]
     \begin{subfigure}[]{\centering
%
%
\definecolor{mycolor1}{rgb}{1.00000,0.00000,1.00000}%
\definecolor{mycolor2}{rgb}{0.00000,1.00000,1.00000}%
\begin{tikzpicture}

\begin{axis}[%
width=0.8\figurewidth,
height=0.65\figureheight,
at={(0\figurewidth,0\figureheight)},
scale only axis,
unbounded coords=jump,
xmin=1,
xmax=10,
xlabel style={font=\color{white!15!black}},
xlabel={\small $K$},
ymin=0,
ymax=0.05,
ylabel style={font=\color{white!15!black}},
ylabel={$\error$},
axis background/.style={fill=white},
        legend columns=3, 
        legend style={
            /tikz/column 3/.style={
                column sep=5pt,
            },
        },
legend style={at={(axis cs:9, -0.0125)},anchor=north east, legend cell align=left, align=right, font=\tiny, draw=none}
]

\addplot [color=blue, dashed, line width=1.0pt, mark=square, mark options={solid, blue}]
  table[row sep=crcr]{%
1	5.55111512312578e-19\\
2	2.7248712730632e-12\\
3	1.10532457479421e-09\\
4	2.44839834351079e-05\\
5	3.68739567791051e-05\\
6	0.0648519921995414\\
};

\addlegendentry{\LemmaUno}

\addplot [color=blue, dashdotted, line width=1.0pt, mark=triangle, mark options={solid, blue}]
  table[row sep=crcr]{%
1	3.33066907387547e-19\\
2	3.05258090924077e-06\\
3	0.0001685453336079\\
4	0.000270457795752786\\
5	0.000196956668183776\\
6	0.109548856718383\\
};
\addlegendentry{ \LZero}

\addplot [color=blue, line width=1.0pt, mark=o, mark options={solid, blue}]
  table[row sep=crcr]{%
1	3.33066907387547e-19\\
2	0.00311369241681917\\
3	0.0119410469315152\\
4	0.0311139637974405\\
5	0.0572354208895527\\
6	0.109548856718383\\
};
\addlegendentry{\LUno,    $\bar{\boldsymbol{\mathsf N}}=\left[1,3,6\right]$}

%
%

\addplot [color=red, dashed, line width=1.0pt, mark=square, mark options={solid, red}]
  table[row sep=crcr]{%
1	2.1094237467878e-18\\
2	1.20301016037772e-10\\
3	4.31910684994818e-12\\
4	7.62011560101472e-10\\
5	8.90791113472567e-10\\
6	7.63599581797464e-08\\
7	8.15373036979254e-07\\
8	0.000131628097297838\\
9	0.00030590871476954\\
10	0.0476207756012005\\
};
\addlegendentry{\LemmaUno}

\addplot [color=red, dashdotted, line width=1.0pt, mark=triangle, mark options={solid, red}]
  table[row sep=crcr]{%
1	6.66133814775094e-19\\
2	2.3385375070356e-06\\
3	1.13119150243288e-05\\
4	0.000226795123952985\\
5	0.000238667195263146\\
6	0.000575130627813619\\
7	0.00080718340885654\\
8	4.67960895728092e-08\\
10	0.109875786410678\\
};
\addlegendentry{ \LZero}

\addplot [color=red, line width=1.0pt, mark=o, mark options={solid, red}]
  table[row sep=crcr]{%
1	6.66133814775094e-19\\
2	0.00224740283490692\\
3	0.0080099528763727\\
4	0.0227129725433205\\
5	0.0396725383666266\\
6	0.0614658618449876\\
7	0.0799560230417046\\
8	0.0943441482142771\\
9	0.102399728189542\\
10	0.109875786410678\\
};
\addlegendentry{\LUno, $\bar{\boldsymbol{\mathsf N}}=\left[1,3,10\right]$}

\addplot [color=black, dashed, line width=1.0pt, mark=square, mark options={solid, black}]
  table[row sep=crcr]{%
1	1.88737914186277e-18\\
2	4.91950924441653e-15\\
3	2.39833260658647e-12\\
4	5.5925783473687e-10\\
5	8.53653697974188e-10\\
6	1.88139970123357e-06\\
7	1.76428800446505e-05\\
8	1.76663131851392e-05\\
9	0.000368285514343347\\
10	0.0468473948985402\\
};
\addlegendentry{\LemmaUno}

\addplot [color=black, dashdotted, line width=1.0pt, mark=triangle, mark options={solid, black}]
  table[row sep=crcr]{%
1	2.22044604925031e-19\\
2	1.38661203947941e-06\\
3	3.92327537278559e-05\\
4	3.90904419628198e-05\\
5	0.000268687530387205\\
6	0.00068525066771479\\
7	4.3701710977951e-05\\
8	0.00263316315317792\\
10	0.104467977401398\\
};
\addlegendentry{\LZero}

\addplot [color=black, line width=1.0pt, mark=o, mark options={solid, black}]
  table[row sep=crcr]{%
1	2.22044604925031e-19\\
2	0.00243707938537063\\
3	0.00799710058841188\\
4	0.0189023855358993\\
5	0.0364724186842356\\
6	0.0524215759368455\\
7	0.0713922762091823\\
8	0.0828835324564031\\
9	0.0911607589585119\\
10	0.104467977401398\\
};
\addlegendentry{\LUno, $\bar{\boldsymbol{\mathsf N}}=\left[1,8,10\right]$}

%
%

\end{axis}
\end{tikzpicture}%
}  
     \end{subfigure}
     \hfill
     \begin{subfigure}[]{\centering
%
%
\begin{tikzpicture}

\begin{axis}[%
width=0.8\figurewidth,
height=0.65\figureheight,
at={(0\figurewidth,0\figureheight)},
scale only axis,
xmin=1,
xmax=6,
xlabel style={font=\color{white!15!black}},
xlabel={\small $K$},
ymin=0,
ymax=0.05,
ylabel style={font=\color{white!15!black}},
ylabel={$\error$},
axis background/.style={fill=white},
        legend columns=3, 
        legend style={
            /tikz/column 3/.style={
                column sep=5pt,
            },
        },
legend style={at={(axis cs:5.4, -0.0175)},anchor=north east, legend cell align=left, align=right, font=\tiny, draw=none}
]

\addplot [color=blue, dashed, line width=1.0pt, mark=square, mark options={solid, blue}]
  table[row sep=crcr]{%
1	2.22044604925031e-18\\
2	1.79833925528783e-14\\
3	2.83740524902814e-12\\
4	0.106755699647375\\
};
\addlegendentry{\LemmaUno}

\addplot [color=blue, dashdotted, line width=1.0pt, mark=triangle, mark options={solid, blue}]
  table[row sep=crcr]{%
1	2.22044604925031e-18\\
2	6.25087404819436e-10\\
3	4.05502374560217e-09\\
4	0.11606693862757\\
};
\addlegendentry{\LZero}

\addplot [color=blue, line width=1.0pt, mark=o, mark options={solid, blue}]
  table[row sep=crcr]{%
1	2.22044604925031e-18\\
2	0.00217522044483404\\
3	0.0130476640349367\\
4	0.11606693862757\\
};
\addlegendentry{\LUno, $\bar{\boldsymbol{\mathsf N}}=\left[4,4,4\right]$}

\addplot [color=red, dashed, line width=1.0pt, mark=square, mark options={solid, red}]
  table[row sep=crcr]{%
1	2.96059473233375e-18\\
2	9.78395142681165e-14\\
3	2.94921766020985e-12\\
4	1.72280065028711e-10\\
5	8.95944781179973e-10\\
6	0.0963325805155264\\
};
\addlegendentry{\LemmaUno}

\addplot [color=red, dashdotted, line width=1.0pt, mark=triangle, mark options={solid, red}]
  table[row sep=crcr]{%
1	4.44089209850063e-18\\
2	7.79936632253373e-11\\
3	1.49854487039228e-08\\
4	9.46309379518834e-05\\
5	5.98959200256748e-08\\
6	0.113513540772132\\
};
\addlegendentry{\LZero}
\addplot [color=red, line width=1.0pt, mark=o, mark options={solid, red}]
  table[row sep=crcr]{%
1	4.44089209850063e-18\\
2	0.00269080682907667\\
3	0.00767417950869976\\
4	0.0167391816766983\\
5	0.0327964165119957\\
6	0.113513540772132\\
};
\addlegendentry{\LUno, $\bar{\boldsymbol{\mathsf N}}=\left[6,6,6\right]$}

\end{axis}
\end{tikzpicture}%
}
         \end{subfigure}
         \vspace{-0.4cm}
       \caption{Frequency average $\ell_1$ recovery error \emph{vs} $K$ for noiseless measurements.  a)  $2$D-UD planar array with  $\NN=\NN_\Csf=\left[1,3,6\right]$, $S_\Csf=10$, $N_\Csf=18$. b) Cubic $3$D-AD array with $4\times  4$ antennas in the cube faces, $\NN_\Csf=\left[2,4,4\right]$, $S_\Csf=10$, $N_\Csf=32$. }
        \LABFIG{fig:error_noiseless}\end{figure}
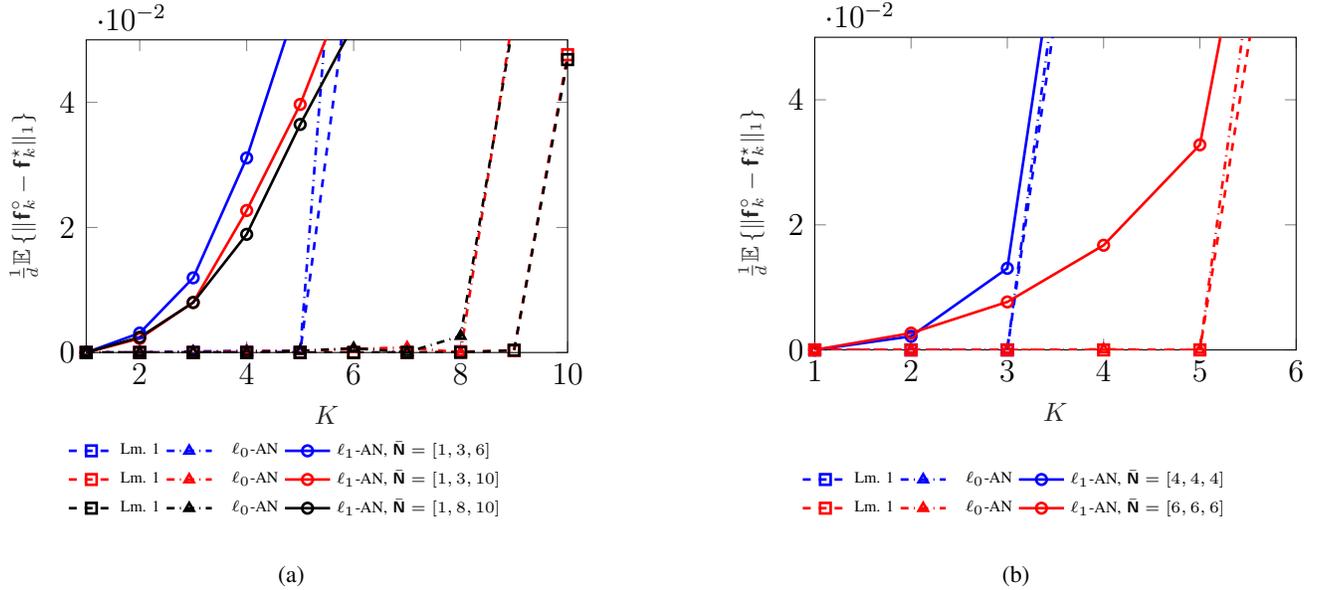

We consider a $2$D-UD physical planar array (uniform structure) and a $3$D-AD physical cubic deployment (non-uniform structure). The $2$D-UD array is characterized by $\NN=\left[1,3,6\right]$ and here the largest underlying uniform subarray coincides with the physical array, i.e., $\NN=\NN_\Csf$, then $S_\Csf=10$ and  $N_\Csf=18$. 
In order to reach all the points in the resolvable region, the $2$D-UD is enlarged with three virtual $2$D-UD: $\bar{\NN}=\left\{\left[1,3,6\right],\left[1,3,10\right],\left[1,8,10\right]\right\}$. 
Similarly the $3$D-AD physical cubic deployment has $4\times 4$ active antennas deployed uniformly in each face of a cube and no antennas are deployed inside the cube. The $3$D-AD has a total of $N=56$ active antennas and its largest underlying $3$D-UD subarray embedded in the non-uniform structure, is given by $\NN_\Csf=[2,4,4]$, with $S_\Csf=10$ and $N_\Csf=32$. Two $3$D-UD virtual deployments are considered for the cubic deployment: $\bar{\NN}=\left[4,4,4\right]$ with $\bar{N}=64$ and $\bar{\NN}=\left[6,6,6\right]$ with $\bar{N}=216$. 

For the aforementioned physical deployments we evaluate, over $1000$ measurements, \emph{i)} the recovery performance and uniqueness of the decomposition proposed in Lemma \ref{PrettyLemma}, \emph{ii)} the resolvable region and error free recovery conditions of the $\ell_0$-AN, and \emph{iii)} the performance of the full $\ell_1$-AN optimization (\ref{l1normopt_convex}). The results are shown in \FIG{fig:error_noiseless}(a-b). While the plots confirm Corollary \ref{CorollaryTh1} resolvability region, they also validate Conjecture \ref{conjec1}. Refer to Table \ref{resolvable} for more detail on the resolvable region. In \FIG{fig:error_noiseless}(a-b), the plots labelled as \LemmaUno~ correspond to  the frequency error between the real frequencies $\mif^\star_{1:K}$ and the recovered $d$-dimensional frequencies $\mif^\circ_{1:K}$ obtained via Vandermonde decomposition  applying 
Algorithm \ref{Alg1} directly to 
$\miS_{\bar{\Xsf}\bar{\Ysf}\bar{\Zsf}}=\sum_{k=1}^K|u_k|^2\mir_{\bar{\NN}}(\mif_k)\mir_{\bar{\NN}}(\mif_k)^\dagger$. \FIG{fig:error_noiseless}(a-b) show that, as stated in Lemma \ref{PrettyLemma},  as long as the condition  $K< \max\{\bar{\Xsf}, \bar{\Ysf}, \bar{\Zsf}\}$ holds (i.e., the  Vandermonde decomposition is unique), the frequencies can be  error-free recovered. The performance of $\ell_0$-AN in \eqref{l0normopt_rank} is evaluated in \FIG{fig:error_noiseless}(a-b) labelled as \LZero. It is shown in all the cases that we have error-free recovery as long as the number of scatters falls within the physical resolvable region (see Table \ref{resolvable}) and $ K<\max\{\bar{\Xsf}, \bar{\Ysf}, \bar{\Zsf}\}$. Finally, the performance of the frequency recovery based on $\ell_1$-AN optimization (\ref{l1normopt_convex}) is also given in \FIG{fig:error_noiseless}. Enlarging the virtual array, in both  uniform and non-uniform settings, improves the recovery performance of the $\ell_1$-AN without additional hardware cost but with an increase of the computational complexity. 

\subsubsection{$\ell_2+\ell_1$-\ac{AN} optimization for noiseless measurements}

Next, we evaluate the performance of the $\ell_2+\ell_1$-\ac{AN} optimization problem in \eqref{eq:SDPnoisy} for noiseless measurements. The results are shown in \FIG{fig:error_noisy_cube}(a) for $\miA=\miA_{\sf U}$, $\NN=\NN_\Csf=\left[1,6,6\right]$,  $\bar{\NN}=\{\left[1,6,6\right],\left[1,6,8\right]\}$ and $K=\{4,5,6\}$. Note that the number of scatters in all cases lay within the lower bound of the resolvable region given in Table \ref{resolvable} and also $K< \max\{\bar{\Xsf}, \bar{\Ysf}, \bar{\Zsf}\}=\bar{\Zsf}$.
It is observed, that as the $\ell_1$-AN error increases the joint $\ell_2+\ell_1$-\ac{AN} approach 
significantly outperforms $\ell_1$-AN, providing a smaller error than the one provided for the virtually enlarged array if $\tau$ is properly optimized. For example for $K=\{5,6\}$ the $\bar{\NN}=\{\left[1,6,8\right]\}$ scenario significantly outperforms the $\ell_1$-AN.



\begin{figure*}
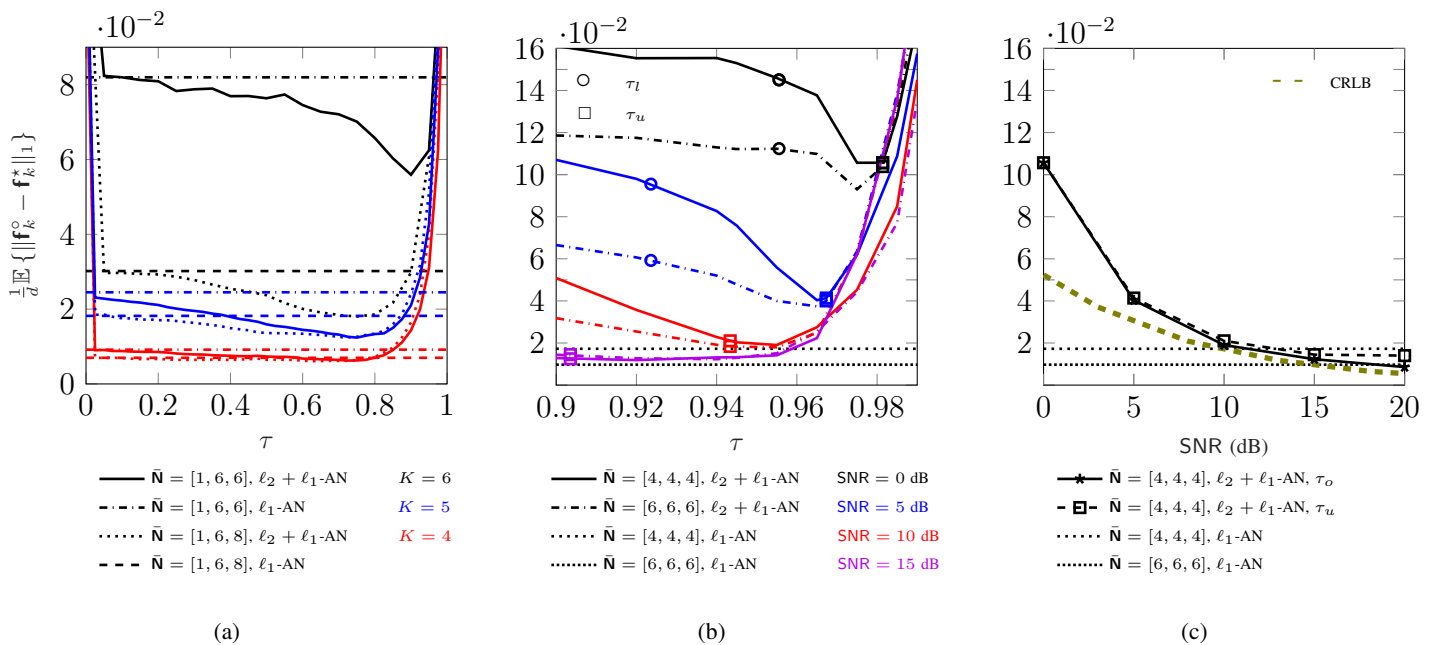

\hspace{-1.4cm}
     \begin{subfigure}[]{\input{figures/Neff6x6_noiseless_L2L1_K4_5_6.tex}}  
     \end{subfigure}\hfill
          \begin{subfigure}[]{\input{figures/Neff4x4x4_noisy_tau.tex}}  
     \end{subfigure}\hfill
     \begin{subfigure}[]{
%
%

\definecolor{mycolor1}{rgb}{0.5,0.5,0.00000}
\definecolor{mycolor2}{rgb}{0.7,0.0,0.9}

\begin{tikzpicture}

\begin{axis}[%
width=0.6\figurewidth,
height=0.7\figureheight,
at={(0\figurewidth,0\figureheight)},
scale only axis,
xmin=0,
xmax=20,
xticklabels={},
xtick={0,5,10,15,20},
extra x ticks ={0,5,10,15,20},
extra x tick labels={$0$,$5$,$10$,$15$,$20$},
xlabel style={font=\color{white!15!black}},
xlabel={\scriptsize{$\snr$ (dB)}},
ymin=0,
ymax=0.16,
yticklabels={},
extra y ticks={0.02,0.04,0.06,0.08,0.1,0.12,0.14,0.16},
extra y tick labels={$2$,$4$,$6$,$8$,$10$,$12$,$14$,$16$},
ylabel style={font=\color{white!15!black}},
axis background/.style={fill=white},
legend style={at={(axis cs:17, -0.035)}, anchor=north east, legend cell align=left, align=left, draw=white!15!black, font=\tiny, draw=none}
]

\addplot [color=black, mark=star, line width=1.0pt]
  table[row sep=crcr]{%
0	0.105648589925056\\
5	0.0403390522842369\\
10	0.0190469975601061\\
15	0.0122973725996679\\
20	0.00857437654291693\\
};
\addlegendentry{$\bar{\boldsymbol{\mathsf N}}=\left[4,4,4\right]$, \LUnoDos, $\tau_o$}

\addplot [color=black, line width=1.0pt, dashed, mark=square, mark options={solid, black}]
  table[row sep=crcr]{%
0	0.105659520255055\\
5	0.0413620806128797\\
10	0.0210404314821107\\
15	0.0144785826678749\\
20	0.0140138870044765\\
};
\addlegendentry{$\bar{\boldsymbol{\mathsf N}}=\left[4,4,4\right]$, \LUnoDos, $\tau_u$}

\addplot [color=black, line width=1.0pt, dotted]
  table[row sep=crcr]{%
0	0.0172572021291106\\
5	0.0172572021291106\\
10	0.0172572021291106\\
15	0.0172572021291106\\
20	0.0172572021291106\\
};
\addlegendentry{$\bar{\boldsymbol{\mathsf N}}=\left[4,4,4\right]$, \LUno}

\addplot [color=black, line width=1.0pt,  densely dotted]
  table[row sep=crcr]{%
0	0.00971414635720909\\
5	0.00971414635720909\\
10	0.00971414635720909\\
15	0.00971414635720909\\
20	0.00971414635720909\\
};
\addlegendentry{$\bar{\boldsymbol{\mathsf N}}=\left[6,6,6\right]$, \LUno}


\addplot [color=mycolor1, line width=2.0pt, dashed, forget plot]
  table[row sep=crcr]{%
0	0.0524\\
3	0.0371295254433105\\
8	0.02087946654148\\
13	0.0117413868842034\\
18	0.00660266706003555\\
23	0.00371295254433106\\
};


%
%
\end{axis}
\node at (0.5,6.75*0.7) {$\cdot10^{-2}$};
\node at (3.7,5.75*0.7) { \textcolor{mycolor1}{- -}\hspace{0.2cm} {\tiny CRLB}};
\end{tikzpicture}%
}
         \end{subfigure}
         \vspace{-0.8cm}
       \caption{Frequency average $\ell_1$ recovery error for \LUnoDos. a) Noiseless measurements in a $2$D-UD planar array with  $\NN=\NN_\Csf=\left[1,6,6\right]$. b) Noisy behavior for a cubic array versus the regularization parameter $\tau$ for $\snr=\{0,5,10,15\}$ dB. c) Noisy behavior for a cubic array  versus $\snr$ for $\tau$ upper bound ($\tau_u$)  and also for the $\tau$ that minimizes the average frequency error ($\tau_o$).} 
\LABFIG{fig:error_noisy_cube}
\end{figure*}

\subsubsection{Noisy measurements}

Finally, we study the noisy scenario frequency recovery performance of \eqref{eq:SDPnoisy}  for different $\snr$ scenarios and the $3$D-AD cubic array.  First, it is explored the behavior with respect to the regularization parameter $\tau$ for $\snr=\{0,5,10,15\}$ (\FIG{fig:error_noisy_cube}(b)). As expected, the regularization parameter $\tau_o$ that minimizes the error depends on $\snr$. Furthermore, for a given $\snr$, the error achieved by $\tau_o$ is almost coincident for the different virtual configurations, i.e., at $\tau_o$, there is no improvement in recovery performance by enlarging the virtual array. From the plots, we can see  that $\tau_u$ well approximate the optimal $\tau_o$ for a large range of $\snr$. Regarding the performance of the different virtual configurations, it is observed that at low $\snr$ values there is a gain provided by the virtual enlarging.  Finally we compare the best recovery performance given by $\tau_o$ with the error performance for the upper bound $\tau_u$, the noiseless $\ell_1$-AN and the \ac{CRLB} in \FIG{fig:error_noisy_cube}(c). The plots show that the performance achieved by the proposed $\ell_2+\ell_1$-\ac{AN} approach are very close to the associated \ac{CRLB}.

 \section{Conclusions} 
 \label{sec:VI}
There is a growing need for multi-dimensional characterization of key performance parameters such as the signal \ac{AoA}s in order to enable the provisioning of future wireless services. In this work, we have shown under which conditions it is possible to recover a set of $K$ $d$-dimensional parameters from a noiseless/noisy linear parametric measurement model that contains a $K$--sparse mixture of the array steering vectors particularized in each of the \ac{AoA}s. Finding the $K$ \ac{AoA}s is formulated in terms of a multi--level Toeplitz matrix $\rank$--minimization problem with a pertinent convex relaxation procedure that ensures the extraction of the relevant parameters in polynomial time.

\vspace{-0.5cm}
\appendices
\section{Proof of Lemma \ref{PrettyLemma}} 
 \label{appPrettyLemma}
 
To prove Lemma \ref{PrettyLemma}, let first observe that for $d=1$ (i.e., $\bar{\Xsf}=\bar{\Ysf}=1$), the result is immediate and well known (see \cite{Yang16,Chi15,Tan14}). For $2\leq d \leq 3$, we will use a result that follows from \cite[Lemma 2]{Yang16} and \cite[Proposition 1]{Gurvits02} that  we will state in a slightly different way which is more useful for our purposes. 

\begin{lemma}
[\hspace{-0.2mm}{\cite[Lemma 2]{Yang16},\cite[Lemma in Proposition 1]{Gurvits02}}]
\label{rottaimmensamenteL}
Let $\miS_{\QQ\TT} \in \mathbb{C}^{\qq\ttt\times \qq\ttt}$ be a  $\qq \times \qq$ \ac{PSD} block Toeplitz matrix with $\rank$ equal to $r<\ttt$, where each block is a $\ttt\times \ttt$ matrix. Then, it exists a matrix $\miG_{\ttt}=[\mig_{1}, \ldots, \mig_{r}]$  and a set of $r$ frequencies
 $ f_j \in \mathbb{T}$, with $j \in [r]$, such that $
\miS_{\QQ\TT} =\sum_{j=1}^r 
\mir_{\qq}(f_j)  \mir_{\qq}(f_j)^\dagger \otimes \mig_{j} \mig_{j}^\dagger
\label{d=2step0} =
\sum_{j=1}^r 
\left (\mir_{{\qq}}(f_j) \otimes \mig_{j}\right)
\left (\mir_{\qq}(f_j) \otimes \mig_{j}\right)^\dagger$, and 
$\miS_{\a\TT} = \miG_{\ttt}\diag(e^{\j2\pi\a f_1}, \ldots, e^{\j2\pi\a f_r})\miG_{\ttt}^\dagger$, where $\miS_{\a\TT}$,  with $-\qq+1 \leq\a \leq \qq-1$, denotes  the matrix on the $\a$-th block of $\miS_{\QQ\TT}$\footnote{Note that while the expression of  $\miS_{\a\TT}$ is not explicitly stated in \cite[Lemma 2]{Yang16},  it is stated as part of the corresponding proof in \cite[Eq. (19)]{Yang16}.}.\QEDwhite
\end{lemma}

We are now ready to prove Lemma \ref{PrettyLemma} for $d=2$, i.e., $\bar{\Xsf}=1$ and $1 <\bar{\Ysf}\leq \bar{\Zsf}$. In this case, $\miS=\miS_{0\bar{\YY}\bar{\ZZ}}$ is a \ac{PSD} $2$-LT matrix  and its upper-left block $\miS_{00\bar{\ZZ}}$ is a \ac{PSD} $1$-LT matrix of dimension $\bar{\Ysf}\bar{\Zsf}$ and $\bar{\Zsf}$, respectively, with $\rank\left\{\miS_{0\bar{\YY}\bar{\ZZ}}\right\}=\rank\left\{\miS_{00\bar{\ZZ}}\right\}=r<\bar{\Zsf}$ by assumption. 
Furthermore, note that since $\miS_{0\bar{\YY}\bar{\ZZ}}$ is a  
\ac{PSD} $2$-LT matrix, then it is also a $\bar{\Ysf} \times \bar{\Ysf}$ \ac{PSD} block Toeplitz matrix  with rank equal to $r<\bar{\Zsf}$ by assumption, 
and with each block being  a $\bar{\Zsf}\times \bar{\Zsf}$ matrix. Hence, applying Lemma \ref{rottaimmensamenteL} with $\qq=\bar{\Ysf}$ and $\ttt=\bar{\Zsf}$,  we have that there exists a matrix $\miG_{\bar{\Zsf}}=[\mig^z_{1}, \ldots, \mig^z_{r}]$  and a set of $r$ frequencies $f^y_j \in \mathbb{T}$, with $j \in [r]$, such that $
\miS_{0\bar{\YY}\bar{\ZZ}}=\sum_{j=1}^r \mir_{\bar{\Ysf}}(f^y_j) \mir_{\bar{\Ysf}}(f^y_j)^\dagger \otimes \mig^z_{j} \mig^{z\dagger}_{j}=\sum_{j=1}^r (\mir_{\bar{\Ysf}}(f^y_j)\otimes  \mig^z_{j})( \mir_{\bar{\Ysf}}(f^y_j) \otimes \mig^{z}_{j})^{\dagger}=\miC_{\bar{\Ysf}\bar{\Zsf}} \miC_{\bar{\Ysf}\bar{\Zsf}}^\dagger
$ and 
\begin{equation}
\begin{split}
\miS_{0\b\bar{\ZZ}}&=\sum_{j=1}^r e^{\j2\pi \b f^y_j}\mig^{z}_{j}\mig^{z\dagger}_{j}=\miG_{\bar{\Zsf}}\diag(e^{\j2\pi\b f^y_1}, \ldots, e^{\j2\pi\b f^y_r})\miG_{\bar{\Zsf}}^\dagger=\miG_{\bar{\Zsf}}\miY^\b\miG_{\bar{\Zsf}}^\dagger
\end{split}
\label{rottac0}
\end{equation}
where $\miS_{0\b\bar{\ZZ}}$, with  $-\Ysf+1\leq\b \leq\Ysf-1$, denotes  
the generic $\b$-block diagonal of $\miS_{0\bar{\YY}\bar{\ZZ}}$, and 
$\miY= \diag({e^{2\pi f^y_1}, e^{2\pi f^y_2}, \ldots,e^{2\pi f^y_r}})$. 
In particular, if we set $\b=0$, we have that $\miS_{00\bar{\ZZ}}=\miG_{\bar{\Zsf}}\miG_{\bar{\Zsf}}^\dagger$.
Furthermore, since by assumption $\miS_{00\bar{\ZZ}}$ is \ac{PSD} $1$-LT matrix of dimension $\bar{\Zsf}$ whose rank $r$ is strictly smaller than $\bar{\Zsf}$, then  $\miS_{00\bar{\ZZ}}$ admits a unique  Vandermonde decomposition of order $r$ (see \cite{Yang16,Chi15,Tan14}), i.e., $\miS_{00\bar{\ZZ}} 
=  \miR_{\bar{\Zsf}}( \mif^z_{1:r}) \miP
 \miR^\dagger_{\bar{\Zsf}}( \mif^z_{1:r})
$, where $\miR_{\bar{\Zsf}}( \mif^z_{1:r})=[\mir_{\bar{\Zsf}}(f^z_1), \ldots, \mir_{\bar{\Zsf}}(f^z_r)]$. Given that $\miS_{00\bar{\ZZ}} =\miG_{\bar{\Zsf}}\miG_{\bar{\Zsf}}^\dagger
=  \miR_{\bar{\Zsf}}( \mif^z_{1:r}) \miP
 \miR^\dagger_{\bar{\Zsf}}( \mif^z_{1:r})$, we can always find a unitary matrix, $\miO_{\bar{\Zsf}}$ such that $\miG_{\bar{\Zsf}}=   \miR_{\bar{\Zsf}}( \mif^z_{1:r}) \miP^{1/2} \miO_{\bar{\Zsf}}
$ and $\miS_{0\b\bar{\ZZ}}$ can be written as $
\miS_{0\b\bar{\ZZ}}= \miR_{\bar{\Zsf}}({ \mif}^z_{1:r}) \miP^{1/2} \miO_{\bar{\Zsf}} \miY^{\b}  \miO_{\bar{\Zsf}} ^\dagger \miP^{1/2} \miR^\dagger_{\bar{\Zsf}}({ \mif}^z_{1:r})$.
Since 
for all $ b$ with  $-\Ysf+1\leq\b \leq\Ysf-1$,  $\miS_{0\b\bar{\ZZ}}$ is a $1$-LT matrix,  using  \cite[Lemma 3]{Yang16} for $d=1$, it follows immediately that the matrix $\miD_{\bar{\Ysf} }=\miO_{\bar{\Zsf}} \miY^{\b} \miO_{\bar{\Zsf}} ^\dagger$ has to be diagonal. Furthermore, 
letting  $\miD_{\bar{\Ysf} }= \miO_{\bar{\Zsf}} \miY  \miO_{\bar{\Zsf}} ^\dagger$, 
it is immediate to verify that $
\miD_{\bar{\Ysf} }^{\b} = \miO_{\bar{\Zsf}} \miY^{\b}  \miO_{\bar{\Zsf}} ^\dagger 
$ and that $\miD_{\bar{\Ysf} }\miD_{\bar{\Ysf} }^\dagger=\miI_r$. 
Therefore, $\miD_{\bar{\Ysf} }$ is a diagonal matrix whose diagonal elements are complex number with  modulo $1$, i.e., $\miD_{\bar{\Ysf} }= \diag(e^{2\pi f^y_1}, \ldots,e^{2\pi f^y_r})$, with ${ f^y_i} \in \mathbb{T}$. Then we have that $\miS_{0\b\bar{\ZZ}}= \miR_{\bar{\Zsf}}({ \mif}^z_{1:r}) \miP^{1/2}
\diag(e^{2\pi f^y_1}, \ldots,e^{2\pi f^y_r})
\miP^{1/2\dagger } \miR^\dagger_{\bar{\Zsf}}({ \mif}^z_{1:r})$ and after some algebraic manipulations, it is easy to show that:
\begin{equation}
\begin{split}
\miS_{0\bar{\YY}\bar{\ZZ}}&=\sum_{j=1}^r p_j
(\mir_{\bar{\Ysf}}(f^y_j)\otimes  \mir_{\bar{\Zsf}}(f^z_j))( \mir_{\bar{\Ysf}}(f^y_j) \otimes \mir_{\bar{\Zsf}}(f^z_j))^\dagger=\miR_{\bar{\Ysf}\bar{\Zsf}}(\mif^{yz}_{1:r}) \miP 
\miR_{\bar{\Ysf}\bar{\Zsf}}^\dagger(\mif^{yz}_{1:r}).
\end{split}
\label{macomesipuoo}
\end{equation}

Now that we have proved the decomposition for $d=2$, we can proceed in proving it for   $d=3$.  
To this end, let us now assume that $\miS=\miS_{\bar{\XX}\bar{\YY}\bar{\ZZ}}$ is a \ac{PSD} $3$-LT matrix, i.e., $1<\bar{\Xsf} \leq \bar{\Ysf}\leq \bar{\Zsf}$ 
and hence the sub-matrix $\miS_{0\bar{\YY}\bar{\ZZ}}$ is a \ac{PSD} $2$-LT matrix of dimension $\bar{\Ysf}\bar{\Zsf}$ and the $\a$-block $\miS_{\a\bar{\YY}\bar{\ZZ}}$  with $-\bar{\Xsf}+1\leq\a\leq\bar{\Xsf}-1$ is also a $2$-LT matrix. By assumption we have that the rank $r$ of $\miS_{\bar{\XX}\bar{\YY}\bar{\ZZ}}$ satisfies
$r=\rank\left\{\miS_{\bar{\XX}\bar{\YY}\bar{\ZZ}}\right\}=
\rank\left\{\miS_{00\bar{\ZZ}}\right\}<\bar{\Zsf}$. Furthermore, since $\miS_{\bar{\XX}\bar{\YY}\bar{\ZZ}}$ is a \ac{PSD} $3$-LT matrix, 
then it is  $\bar{\Xsf} \times \bar{\Xsf}$ \ac{PSD} block Toeplitz matrix whose blocks are $2$-LT matrices. Therefore, we can apply Lemma \ref{rottaimmensamenteL},
with $\qq=\bar{\Xsf}$ and $t=\bar{\Ysf}\bar{\Zsf}$, from which it follows that it exists a matrix $\miG_{\bar{\Ysf}\bar{\Zsf}}=[\mig^{yz}_{1}, \ldots, \mig^{yz}_{r}]$  and $f^x_j \in \mathbb{T}$, with $j \in [r]$, such that $
\miS_{\bar{\XX}\bar{\YY}\bar{\ZZ}}=\sum_{j=1}^r 
\mir_{\bar{\Xsf}}(f^x_j)  \mir_{\bar{\Xsf}}(f^x_j)^\dagger \otimes \mig^{yz}_{j} \mig^{yz\dagger}_{j}=\miC_{\bar{\Xsf}\bar{\Ysf}\bar{\Zsf}}\miC_{\bar{\Xsf}\bar{\Ysf}\bar{\Zsf}}^\dagger$, and the generic $\a$-block of  $\miS_{\bar{\XX}\bar{\YY}\bar{\ZZ}}$ with  $-\Xsf+1\leq\a \leq\Xsf-1$, can be written as 
\begin{equation}
\begin{split}
\miS_{\a\bar{\YY}\bar{\ZZ}}&=\sum_{j=1}^r e^{\j2\pi \a f^x_j}\mig^{yz}_{j}\mig^{yz\dagger}_{j}=\miG_{\bar{\Ysf}\bar{\Zsf}}\diag(e^{\j2\pi\a f^x_1}, \ldots, e^{\j2\pi\a f^x_r})\miG_{\bar{\Ysf}\bar{\Zsf}}^\dagger=\miG_{\bar{\Ysf}\bar{\Zsf}}\miX^\a\miG_{\bar{\Ysf}\bar{\Zsf}}^\dagger
\label{ec_ayz}
\end{split}
\end{equation}
with $\miX= \diag(e^{2\pi f^x_1}, \ldots,e^{2\pi f^x_r}) $.
Furthermore,
its $\bar{\Ysf}\bar{\Zsf} \times \bar{\Ysf}\bar{\Zsf}$ upper-left corner, say $\miS_{0\bar{\YY}\bar{\ZZ}}$, is a $2$-LT matrix, composed by the first $\bar{\Ysf}\bar{\Zsf}$ rows and the first $\bar{\Ysf}\bar{\Zsf}$ columns of $\miS_{\bar{\XX}\bar{\YY}\bar{\ZZ}}$, whose  $\bar{\Zsf} \times \bar{\Zsf}$  upper-left block of $\miS_{0\bar{\YY}\bar{\ZZ}}$, say $ \miS_{00\bar{\ZZ}}$, is a $1$-LT matrix.
It is immediate to prove that $\rank\left\{\miS_{0\bar{\YY}\bar{\ZZ}}\right\} 
=\rank \left\{\miS_{00\bar{\ZZ}} \right\}=r< \bar{\Zsf} $. In fact 
$\rank \left\{\miS_{00\bar{\ZZ}} \right\}= r \leq \rank\left\{\miS_{0\bar{\YY}\bar{\ZZ}}\right\}\leq \rank\left\{\miS_{\bar{\XX}\bar{\YY}\bar{\ZZ}}\right\}=r$.
Hence, the assumption of Lemma~\ref{PrettyLemma}, proved previously for $d=2$, are satisfied by $\miS_{0\bar{\YY}\bar{\ZZ}}$. Therefore, $\miS_{0\bar{\YY}\bar{\ZZ}}$ admits a Vandermonde decomposition of order $r$, i.e $
\miS_{0\bar{\YY}\bar{\ZZ}} =\miR_{\bar{\Ysf}\bar{\Zsf}}(\mif^{yz}_{1:r}) \miP
\miR^\dagger_{\bar{\Ysf}\bar{\Zsf}}(\mif^{yz}_{1:r})$.
Following  steps and arguments very similar to the case $d=2$ and 
using again \cite[Lemma 3]{Yang16} to the generic  $2$-LT matrix  block  $\miS_{\a\bar{\YY}\bar{\ZZ}}$, it is straight forward, after some algebraic manipulations, to verify that:
\begin{equation}
\label{macomesipuo}
\miS_{\bar{\XX}\bar{\YY}\bar{\ZZ}}=
\sum_{j=1}^{r} p_j
\mir_{\bar{\NN}}(\mif_{j})
\mir_{\bar{\NN}}(\mif_{j})^\dagger=\miR_{\bar{\NN}}(\mif_{1:r})\miP\miR_{\bar{\NN}}^\dagger(\mif_{1:r}).
\end{equation}

In order complete the proof, both for $\dd=2$ and $\dd=3$, we need to show uniqueness of the Vandermonde decomposition given in \eqref{macomesipuo}. 
To do this let assume that it exists a second decomposition such that $
\miS_{\bar{\XX}\bar{\YY}\bar{\ZZ}}=
\miR_{\bar{\NN}}(\mif'_{1:r})\miP'\miR_{\bar{\NN}}^\dagger(\mif'_{1:r})$, where, as in \eqref{macomesipuo}, $\miP'= \diag\left(p'_1,\ldots, p'_{r}\right)$ with $p'_k>0$ $k\in[r]$ and $\mif'_{k}\in  \mathbb{T}^\dd$ representing distinct vectors.
Then, it follows that $
\miR_{\bar{\NN}}(\mif'_{1:r})\miP'\miR_{\bar{\NN}}^\dagger(\mif'_{1:r})=\miR_{\bar{\NN}}(\mif_{1:r})\miP\miR_{\bar{\NN}}^\dagger(\mif_{1:r})$, from which we have that $
\miR_{\bar{\NN}}(\mif'_{1:r})=
\miR_{\bar{\NN}}(\mif_{1:r})\miP^{\frac{1}{2}} \miU \miP^{'-\frac{1}{2}}=\miR_{\bar{\NN}}(\mif_{1:r})\miB$, where $\miU$ is a unitary matrix and $\miB=\miP^{ \frac{1}{2}} \miU \miP^{'-\frac{1}{2}}$. As a result, for any $k\in \{1, \ldots,r\}$, $\mir_{\bar{\NN}}(\mif'_{k})$
can be written  as linear combination of the steering vectors  
$\left \{\mir_{\bar{\NN}}(\mif_{j})\right \}_{j=1}^r$, i.e.:
\begin{equation}
\label{macomesipuo4}
\mir_{\bar{\NN}}(\mif'_{k}) =
\miR_{\bar{\NN}}(\mif_{1:r})\miB_{(k)} =\sum_{j=1}^r \mir_{\bar{\NN}}(\mif_{j})b_{jk}.
\end{equation}

This implies that the  $r+1$ steering vectors $\left \{ \left \{\mir_{\bar{\NN}}(\mif_{j}) \right \}_{j=1}^r \bigcup \mir_{\bar{\NN}}(\mif'_{k})\right\} $ are linearly dependent.  

Next, denoting by $\miR_{\bar{\NN}}(\mif_{1:r}, \mif'_{k} )=[\miR_{\bar{\NN}}(\mif_{1:r}),\mir_{\bar{\NN}}(\mif'_{k})]$ the ${\bar{N}}\times (r+1)$ matrix, we have that $\miR_{\bar{\NN}}(\mif_{1:r}, \mif'_{k} )= 
\miR_{\bar{\Xsf}}(\mif^x_{1:r}, f'^{x}_{k} ) \odot 
\miR_{\bar{\Ysf}}(\mif^y_{1:r}, f'^{y}_{k} ) \odot  \miR_{\bar{\Zsf}}(\mif^z_{1:r}, f'^{z}_{k} )$ from which it follows that the 
rank of $\miR_{\bar{\NN}}(\mif_{1:r}, \mif'_{k} )$ satisfies the following inequality (see \cite[Lemma 1]{Jiang01}):
\begin{equation}
\begin{split}
\rank \left \{\miR_{\bar{\NN}}(\mif_{1:r}, \mif'_{k} ) \right\}&\geq  
\min \left \{ 
\min \left \{r_{\bar{\Zsf}}  + r_{\bar{\Ysf}} -1 +(a_z+a_y), r+1 \right \}+
 r_{\bar{\Xsf}} -1 +a_x, r+1 \right \} 
\\
&= \min\left \{r+a_z+a_y+a_x, r+1 \right \} 
\end{split}
\label{macomesipuo6}
\end{equation}
where  $r_{\bar{\Zsf}}=\rank 
\left \{\miR_{\bar{\Zsf}}(\mif^z_{1:r} ) \right\}$,
$r_{\bar{\Ysf}}=\rank 
\left \{\miR_{\bar{\Ysf}}(\mif^y_{1:r} ) \right\}$,
$r_{\bar{\Xsf}}=\rank 
\left \{\miR_{\bar{\Xsf}}(\mif^x_{1:r}) \right\}$, while 
$a_\alpha=1$ if $f^{'\alpha}_{k} \notin \left \{f^\alpha_{j} \right\}_{j=1}^r$ and $\beta^\alpha>r$ and zero otherwise, with $\alpha \in \left \{x,y,z \right \}$ and $\beta^x=\bar{\Xsf}$, $\beta^y=\bar{\Ysf}$ and $\beta^z=\bar{\Zsf}$.
From \eqref{macomesipuo6}, it follows that in order for \eqref{macomesipuo4} to be valid for any $k\in \{1, \ldots,r\}$, it is necessary that for any $\left \{\mif'_{k} \right\}_{k=1}^r \subset \left \{\mif_{j} \right\}_{j=1}^r$, i.e $a_z+a_y+a_x=0$. By similar argument we can prove that 
$\left \{\mif'_{k} \right\}_{k=1}^r \supset \left \{\mif_{j} \right\}_{j=1}^r$, from which it follows immediately that 
also $p'_k=p_k$ for all $k\in \{1, \ldots,r\}$ and consequently the decomposition is unique. Hence if  Vandermonde  decomposition  in \eqref{macomesipuoo} for $\dd=2$ or \eqref{macomesipuo} for $\dd=3$ exists, it is unique.  This completes the proof.

\vspace{-0.5cm}
\section{Proof of Theorem \ref{theoremApi}} 
 \label{appAIW0}

In the following, we will prove that given the assumption of Theorem \ref{theoremApi}, i.e., $K<\max\left\{\bar{\Xsf},\bar{\Ysf},\bar{\Zsf}\right\}$, denoting by $\big(r^\circ,\mis^\circ,\miS^\circ_{\bar{\XX}\bar{\YY}\bar{\ZZ}}\big)$ 
the optimal solution  to \eqref{l0normopt_rank} is 
given by
\begin{equation}
\big(r^\circ,\mis^\circ,\miS^\circ_{\bar{\XX}\bar{\YY}\bar{\ZZ}}\big) = \bigg(K,\sum_{k=1}^Ku^\star_k\mir_{\bar{\NN}}(\mif^\star_k),\sum_{k=1}^K|u^\star_k|^ 2\mir_{\bar{\NN}}(\mif^\star_k)\mir_{\bar{\NN}}(\mif^\star_k)^\dagger\bigg)
\label{optsolappedix}
\end{equation}
and is  unique in terms of $\mis^\circ=\mis^\star=\sum_{k=1}^Ku^\star_k\mir_{\bar{\NN}}(\mif^\star_k)$ and in terms of the frequencies $\mif^\star_{1:K}$  identified using Algorithm \ref{Alg1}. 


To prove \eqref{optsolappedix}, let first observe that $\bigg(K,\sum_{k=1}^Ku^\star_k\mir_{\bar{\NN}}(\mif^\star_k),\sum_{k=1}^K|u^\star_k|^2\mir_{\bar{\NN}}(\mif^\star_k)\mir_{\bar{\NN}}(\mif^\star_k)^\dagger\bigg)$ is a feasible solution  for \eqref{l0normopt_rank}, therefore $
\rank \left \{ \miS^\circ_{\bar{\XX}\bar{\YY}\bar{\ZZ}}\right\}\triangleq r^\circ\leq K<\max\{\bar{\Xsf}, \bar{\Ysf}, \bar{\Zsf}\}$.

On the other hand, it will be proved next, \emph{1)} that $K \leq r^\circ$, \emph{2)} that $\mis^\circ$ is unique and equal to $
\mis^\star= \sum_{k=1}^Ku^\star_k\mir_{\bar{\NN}}(\mif^\star_k)$ and finally \emph{3)}  that $\miS^\circ_{\bar{\XX}\bar{\YY}\bar{\ZZ}}$ has also a unique decomposition in terms of the $\mif^\star_{1:K}$  identified using Algorithm \ref{Alg1}.

In the following we assume $d=3$, nevertheless the proof can be easily particularized for $d\leq 2$ and $d> 3$. 
By formulation of \eqref{l0normopt_rank},  the optimal solution $\miS^\circ_{\bar{\XX}\bar{\YY}\bar{\ZZ}}$ belongs to $\mathcal{T}_{\bar{\Xsf}\bar{\Ysf}\bar{\Zsf}} \subseteq \C^{\bar{N} \times \bar{N}} $ which  denotes the set of all 
\ac{PSD} $3$-LT matrices of dimension $ \bar{\Xsf}\bar{\Ysf}\bar{\Zsf}$,
with  canonical ordered structure. Consequently, its $\bar{\Ysf}\bar{\Zsf}\times \bar{\Ysf}\bar{\Zsf}$ upper block $\miS^\circ_{0\bar{\YY}\bar{\ZZ}}$ is a \ac{PSD} $2$-LT matrix and the  $\bar{\Zsf}\times \bar{\Zsf}$ upper block  $\miS^\circ_{00\bar{\ZZ}}$ is a \ac{PSD} $1$-LT matrix. We denote the $\rank$ of each of these matrices as $r^\circ=\rank\left\{\miS^\circ_{\bar{\XX}\bar{\YY}\bar{\ZZ}}\right\}$, $r^\circ_{\bar{\Ysf}\bar{\Zsf}}=\rank\left\{\miS^\circ_{0\bar{\YY}\bar{\ZZ}}\right\}$ and $r^\circ_{\bar{\Zsf}}=\rank\left\{\miS^\circ_{00\bar{\ZZ}}\right\}$ and we have that $r^\circ_{\bar{\Zsf}}\leq r^\circ_{\bar{\Ysf}\bar{\Zsf}}\leq r^\circ <\max\{\bar{\Xsf}, \bar{\Ysf}, \bar{\Zsf}\}= \bar{\Zsf}$. By Lemma \ref{rottaimmensamenteL}, we have that it exists a matrix 
$\miG_{ \bar{\Ysf}\bar{\Zsf}}=[\mig^{yz}_{1}, \ldots, \mig^{yz}_{r^\circ}]$  and $f^{x\circ}_j \in \mathbb{T}$, with $j \in [r^\circ]$ such that $
\miS^\circ_{\bar{\XX}\bar{\YY}\bar{\ZZ}}=\sum_{j=1}^{r^\circ} (\mir_{\bar{\Xsf}}(f^{x\circ}_j)\otimes  \mig^{yz}_{j})( \mir_{\bar{\Xsf}}(f^{x\circ}_j) \otimes \mig^{yz}_{j})^\H=\miC_{\bar{\Xsf}\bar{\Ysf}\bar{\Zsf}} \miC_{\bar{\Xsf}\bar{\Ysf}\bar{\Zsf}}^\dagger$, where 
\begin{eqnarray}
\miC_{\bar{\Xsf}\bar{\Ysf}\bar{\Zsf}}&=&[\mir_{\bar{\Xsf}}(f^{x\circ}_1)\otimes  \mig^{yz}_{1},\ldots,\mir_{\bar{\Xsf}}(f^{x\circ}_{r^\circ})\otimes  \mig^{yz}_{r^\circ}] =\miR_{\bar{\Xsf}}(f^{x\circ}_{1:r^\circ}) \odot \miG_{\bar{\Ysf}\bar{\Zsf}}.
\label{machec0}
\end{eqnarray}
Also denoting $\miS^\circ_{\a\bar{\YY}\bar{\ZZ}}$ as the $\a$-th block of 
$\miS^\circ_{\bar{\XX}\bar{\YY}\bar{\ZZ}}$ for  $-\bar{\Xsf}+1\leq \a\leq\bar{\Xsf}-1$ we have that $\miS^\circ_{\a\bar{\YY}\bar{\ZZ}}= \sum_{j=1}^{r^\circ} e^{\j2\pi \a f^{x\circ}_j} \mig^{yz}_{j}\mig ^{yz\dagger}_{j} =\miG_{\bar{\Ysf}\bar{\Zsf}}\miX^\a\miG_{\bar{\Ysf}\bar{\Zsf}}^\dagger$ with $\miX= \diag({e^{\j2\pi f^{x\circ}_1}, e^{\j2\pi f^{x\circ}_2}, \ldots,e^{\j2\pi f^{x\circ}_{r^\circ}}})$. Similarly, applying Lemma \ref{rottaimmensamenteL} to $\miS^\circ_{0\bar{\YY}\bar{\ZZ}}$, it exists a matrix $\miG_{ \bar{\Zsf}}=[\mig^{z}_{1}, \ldots, \mig^{z}_{r^\circ_{\bar{\Ysf}\bar{\Zsf}}}]$  and $f^{y\circ}_j \in \mathbb{T}$, with $j \in [r^\circ_{\bar{\Ysf}\bar{\Zsf}}]$, such that $\miS^\circ_{0\bar{\YY}\bar{\ZZ}}=\sum_{j=1}^{r^\circ_{\bar{\Ysf}\bar{\Zsf}}} 
(\mir_{\bar{\Ysf}}(f^{y\circ}_j)\otimes  \mig^{z}_{j})( \mir_{\bar{\Ysf}}(f^{y\circ}_j) \otimes \mig^{z}_{j})^\H =\miC_{\bar{\Ysf}\bar{\Zsf}} \miC_{\bar{\Ysf}\bar{\Zsf}}^\dagger$
where 
\begin{eqnarray}
\miC_{\bar{\Ysf}\bar{\Zsf}} &=&
[\mir_{\bar{\Ysf}}(f^{y\circ}_1)\otimes  \mig^{z}_{1},\ldots,\mir_{\bar{\Ysf}}(f^{y\circ}_{r^\circ_{\bar{\Ysf}\bar{\Zsf}}})\otimes  \mig^{z}_{r^\circ_{\bar{\Ysf}\bar{\Zsf}}}]=\miR_{\bar{\Ysf}}(f^{y\circ}_{1:r^\circ_{\bar{\Ysf}\bar{\Zsf}}}) \odot \miG_{\bar{\Zsf}}. 
\label{machec1}
\end{eqnarray}
Similarly, the generic $\b$-block of $\miS^\circ_{0\bar{\YY}\bar{\ZZ}}$ for  $-\bar{\Ysf}+1\leq\b \leq\bar{\Ysf}-1$ is given by $
\miS^\circ_{0\b\bar{\ZZ}}=\sum_{j=1}^{r^\circ_{\bar{\Ysf}\bar{\Zsf}}} e^{\j2\pi \b f^{y\circ}_j} \mig^{z}_{j}\mig ^{z\dagger}_{j} =\miG_{\bar{\Zsf}}\miY^\b\miG_{\bar{\Zsf}}^\dagger$ with $\miY= \diag(e^{\j2\pi f^{y\circ}_1}, e^{\j2\pi f^{y\circ}_2}, \ldots,e^{\j2\pi f^{y\circ}_{r^\circ_{\bar{\Ysf}\bar{\Zsf}}}})$. Furthermore, $\miS^\circ_{00\bar{\ZZ}}$ admits a Vandermonde decomposition of order $r^\circ_{\bar{\Zsf}} < \bar{\Zsf}$, 
$\miS^\circ_{00\bar{\ZZ}}=\miC_{\bar{\Zsf}}\miC_{\bar{\Zsf}}^\dagger=  \miR_{\bar{\Zsf}}(f^{z\circ}_{1:r^\circ_{\bar{\Zsf}}}) \miP
 \miR_{\bar{\Zsf}}(f^{z\circ}_{1:r^\circ_{\bar{\Zsf}}})^\dagger$, where
 $$\miC_{\bar{\Zsf}}=[\mir_{\bar{\Zsf}}(f^{z\circ}_1)p^{1/2}_1,\ldots,\mir_{\bar{\Zsf}}(f^{z\circ}_{r^\circ_{\bar{\Zsf}}})p^{1/2}_{r^\circ_{\bar{\Zsf}}}]=\miR_{\bar{\Zsf}}(f^{z\circ}_{1:r^\circ_{\bar{\Zsf}}}) \miP^{\frac{1}{2}}$$
 Combining $\miS^\circ_{00\bar{\ZZ}}=\miC_{\bar{\Zsf}}\miC_{\bar{\Zsf}}^\dagger$ with $\miS^\circ_{0\b\bar{\ZZ}}$ evaluated for $\b=0$,
and  $\miS^\circ_{0\bar{\YY}\bar{\ZZ}} =\miC_{\bar{\Ysf}\bar{\Zsf}} \miC_{\bar{\Ysf}\bar{\Zsf}}^\dagger$ with $\miS^\circ_{\a\bar{\YY}\bar{\ZZ}}$ evaluated for $\a=0$
we have that $ 
\miG_{\bar{\Zsf}}\miG_{\bar{\Zsf}}^\dagger=\miC_{\bar{\Zsf}}\miC_{\bar{\Zsf}}^\dagger$ and
$\miG_{\bar{\Ysf}\bar{\Zsf}}\miG_{\bar{\Ysf}\bar{\Zsf}}^\dagger=\miC_{\bar{\Ysf}\bar{\Zsf}}\miC_{\bar{\Ysf}\bar{\Zsf}}^\dagger$
and therefore, we can always find  a 
$r^\circ_{\bar{\Zsf}}\times r^\circ_{\bar{\Ysf}\bar{\Zsf}}$ unitary matrix, $\miO_{\bar{\Zsf}}$ and a $r^\circ_{\bar{\Ysf}\bar{\Zsf}}\times r^\circ$ unitary matrix, $\miO_{\bar{\Ysf}\bar{\Zsf}}$, such that 
\begin{eqnarray}
\miG_{\bar{\Zsf}}=\miC_{\bar{\Zsf}}\miO_{\bar{\Zsf}}, \qquad  \miG_{\bar{\Ysf}\bar{\Zsf}}=\miC_{\bar{\Ysf}\bar{\Zsf}}\miO_{\bar{\Ysf}\bar{\Zsf}}, 
\label{musica}
\end{eqnarray}

Replacing  \eqref{musica} in \eqref{machec1}  and \eqref{machec0} we have that 
\begin{equation}
\begin{split}
\miC_{\bar{\Xsf}\bar{\Ysf}\bar{\Zsf}}&=
\miR_{\bar{\Xsf}}(f^{x\circ}_{1:r^\circ})\odot\bigg(\miR_{\bar{\Ysf}}(f^{y\circ}_{1:r^\circ_{\bar{\Ysf}\bar{\Zsf}}}) \odot\big(\miR_{\bar{\Zsf}}(f^{z\circ}_{1:r^\circ_{\bar{\Zsf}}})\miP^{\frac{1}{2}} \miO_{\bar{\Zsf}}\big)\bigg)\miO_{\bar{\Ysf}\bar{\Zsf}}=\bigg[\mic_1^{xyz},\dots,\mic_{r^\circ}^{xyz}\bigg],\label{Cyz}
\end{split}
\end{equation}
where $
\mic_j^{xyz}=
\mir_{\bar{\Xsf}}(f^{x\circ}_j) \otimes
\sum_{i=1}^{r^\circ_{\bar{\Ysf}\bar{\Zsf}}}o^{yz}_{ji}
\mir_{\bar{\Ysf}}(f^{y\circ}_i) \otimes\sum_{l=1}^{r^\circ_{\bar{\Zsf}}}p^{\frac{1}{2}}_l o^{z}_{il}\mir_{\bar{\Zsf}}(f^{z\circ}_l)$ for $j\in[r^\circ]$. Due to Schur complement lemma, the optimal signal $\mis^\circ$ that satisfies constraint $\begin{bmatrix}\miS^\circ_{\bar{\XX}\bar{\YY}\bar{\ZZ}}& \mis^\circ\\ \mis^{\circ\H} & r^\circ\end{bmatrix}\succeq 0$ in \eqref{l0normopt_rank} has to be in the span of $\miS^{\circ}_{\bar{\XX}\bar{\YY}\bar{\ZZ}}$. Therefore, there exists a set of coefficients $[\alpha_1^\circ, \ldots, \alpha_{r^\circ}^\circ]$ such that:
\begin{equation}
\begin{split}
\mis^\circ&=\sum_{j=1}^{r^\circ}\alpha_j^\circ \mic_j^{xyz}=\sum_{j=1}^{r^\circ}\alpha_j^\circ\mir_{\bar{\Xsf}}(f^{x\circ}_j) \otimes\sum_{i=1}^{r^\circ_{\bar{\Ysf}\bar{\Zsf}}}o^{yz}_{ji}\mir_{\bar{\Ysf}}(f^{y\circ}_i) \otimes\sum_{l=1}^{r^\circ_{\bar{\Zsf}}}p^{\frac{1}{2}}_l o^{z}_{il}\mir_{\bar{\Zsf}}(f^{z\circ}_l)\\
&=\sum_{j=1}^{r^\circ}\sum_{i=1}^{r^\circ_{\bar{\Ysf}\bar{\Zsf}}}\sum_{l=1}^{r^\circ_{\bar{\Zsf}}}\alpha_j^\circ o^{yz}_{ji}p^{\frac{1}{2}}_l o^{z}_{il}\mir_{\bar{\Xsf}}(f^{x\circ}_j) \otimes \mir_{\bar{\Ysf}}(f^{y\circ}_i) \otimes\mir_{\bar{\Zsf}}(f^{z\circ}_l),
\end{split}
\label{so}
\end{equation}
where we have $r^\circ r^\circ_{\bar{\Ysf}\bar{\Zsf}} r^\circ_{\bar{\Zsf}}$ $\bar{N}$-dimensional  vectors $\mir_{\bar{\Xsf}}(f^{x\circ}_j) \otimes\mir_{\bar{\Ysf}}(f^{y\circ}_i) \otimes\mir_{\bar{\Zsf}}(f^{z\circ}_l)$ with $j\in[r^\circ]$, $i\in[r^\circ_{\bar{\Ysf}\bar{\Zsf}}]$ and $l\in[r^\circ_{\bar{\Zsf}}]$ that are necessarily linearly dependent given that $\rank\left\{\miS^\circ_{\bar{\XX}\bar{\YY}\bar{\ZZ}}\right\}=r^\circ$. We define $\mathcal{M}$ as the set of indexes $(j_m,i_m,l_m)$ such that the vectors  $\mir_{\bar{\Xsf}}(f^{x\circ}_{j_m}) \otimes\mir_{\bar{\Ysf}}(f^{y\circ}_{i_m}) \otimes\mir_{\bar{\Zsf}}(f^{z\circ}_{l_m})$ are linearly independent, note that $|\mathcal{M}|=r^\circ$ and $m\in[r^\circ]$. Then, we can rewrite $\mis^\circ$ as a linear combination of those $r^\circ$ linearly independent $\mir_{\bar{\Xsf}}(f^{x\circ}_{m}) \otimes\mir_{\bar{\Ysf}}(f^{y\circ}_{m}) \otimes\mir_{\bar{\Zsf}}(f^{z\circ}_{m})$ vectors, now indexed in $m\in[r^\circ]$, $
\mis^\circ =\sum_{m=1}^{r^\circ}u^\circ_{m}\mir_{\bar{\Xsf}}(f^{x\circ}_{m}) \otimes\mir_{\bar{\Ysf}}(f^{y\circ}_{m}) \otimes\mir_{\bar{\Zsf}}(f^{z\circ}_{m})=\sum_{m=1}^{r^\circ}u^\circ_{m}\mir_{\bar{\NN}}(\mif^{\circ}_{m})$. It follows then that $\mis^\circ$ can be expressed as a linear combination of $r^\circ$  steering vectors   with a Kronecker ordering $\bar{\Zsf}\to\bar{\Ysf}\to\bar{\Xsf}$ given by the  nesting ordering of $\miS_{\bar{\XX}\bar{\YY}\bar{\ZZ}}$ which by formulation of \eqref{l0normopt_rank} is assumed to have canonical ordered structure.

Next, by being solution of \eqref{l0normopt_rank}, $\mis^\circ$ satisfies $\miA\mis^\circ=\miA\mis^\star$. Hence, using the injectivity property of  $\miA\in\mathcal{A}^{(K)}$ (see Definition \ref{setinjective}), we will prove that, if $\mif^\star_{1:K}$ satisfies Assumption \ref{A1}, $r^\circ=K$ and that it exists a unique set of $r^\circ=K$ frequencies that gives as linear combination $\mis^\circ$  and this set coincides with $\mif^\star_{1:K}$.
To this end let us recall a classical result: 
\begin{lemma}
\label{lemmatrivia}
Given a vector $\mis^\star= \sum_{k=1}^Ku^\star_k\mir_{\bar{\NN}}(\mif^\star_k)$, with 
$\mif^\star_{1:K}$ satisfying 
Assumption \ref{A1},  and given a vector $\mis^{\diamond}=
\sum_{j=1}^{r^{\diamond}}u^{\diamond}_k\mir_{\bar{\NN}}(\mif^{\diamond}_j)$ with $\mif^{\diamond}_{1:r^{\diamond}}$ not necessarily satisfying Assumption \ref{A1}, 
for any given matrix $\miA\in\mathcal{A}^{(K)}$ (see Definition \ref{setinjective}), if $\miA\mis^{\diamond}=\miA\mis^\star$, 
there exists a unique set of frequencies that gives as linear combination $\mis^{\diamond}$, provided that $r^{\diamond} \leq K$,
and this set coincides with $\mif^\star_{1:K}$.
\end{lemma}
\begin{proof}
Lemma \ref{lemmatrivia} is a variation of classical and widely known results, in the context of vector sparse representation.
\end{proof}
Applying Lemma \ref{lemmatrivia}, it follows immediately that $\mis^\circ=\mis^\star$, $r^\circ=K$, $\mif^\circ_{1:r^\circ}=\mif^\star_{1:K}$. From this, it also follows immediately that the span of  $\miS^\circ_{\bar{\XX}\bar{\YY}\bar{\ZZ}}$
admits as generating vectors the set of steering vectors with nesting ordering $\bar{\Zsf}\to\bar{\Ysf}\to\bar{\Xsf}$ associated to the set of $K$  frequencies $\mif^\star_{1:K}$, i.e., the columns of $\miR_{\bar{\NN}}(\mif^\star_{1:K})$. Hence, to complete the theorem we need now to prove that  $\miS^\circ_{\bar{\XX}\bar{\YY}\bar{\ZZ}}$ admits a unique Vandermonde decomposition that can be obtained via Algorithm 
\ref{Alg1} in order to uniquely identify the frequencies $\mif^\star_{1:K}$.

With this aim, let us consider the first $\bar{\Zsf}$ components of 
$\mis^\circ=\mis^\star$. They identify a $\bar{\Zsf}$-dimensional vector  that we  denote by $\miA_{\bar{\Zsf}}\mis^\circ$. Analogously, we denote by $\miA_{\bar{\Zsf}}\mis^\star$ the first $\bar{\Zsf}$ components of 
$\mis^\star= \sum_{k=1}^Ku^\star_k\mir_{\bar{\NN}}(\mif^\star_k)$. From \eqref{so}, we have that $\miA_{\bar{\Zsf}}\mis^\circ=\sum_{l=1}^{r^\circ_{\bar{\Zsf}}}\mir_{\bar{\Zsf}}(f^{z\circ}_l)\sum_{j=1}^{r^\circ}\sum_{i=1}^{r^\circ_{\bar{\Ysf}\bar{\Zsf}}} \alpha_j^\circ o^{yz}_{ji}p^{\frac{1}{2}}_l o^{z}_{il}$ and also since
$\miA_{\bar{\Zsf}}\mis^\circ=\miA_{\bar{\Zsf}}\mis^\star$, we also have that $
\sum_{l=1}^{r^\circ_{\bar{\Zsf}}}\mir_{\bar{\Zsf}}(f^{z\circ}_l)\sum_{j=1}^{r^\circ}\sum_{i=1}^{r^\circ_{\bar{\Ysf}\bar{\Zsf}}} \alpha_j^\circ o^{yz}_{ji}p^{\frac{1}{2}}_l o^{z}_{il} 
=\sum_{k=1}^K\mir_{\bar{\Zsf}}(f^{z\star}_k)u^\star_k $ which cannot be true unless $K=r^\circ=r^\circ_{\bar{\Zsf}}$. Finally, since we had that $r^\circ\geq r^\circ_{\bar{\Ysf}\bar{\Zsf}}\geq r^\circ_{\bar{\Zsf}}$, it is true then that $K=r^\circ=r^\circ_{\bar{\Ysf}\bar{\Zsf}}=r^\circ_{\bar{\Zsf}}$. 
From this it follows that 
$\rank\left\{\miS^\circ_{\bar{\XX}\bar{\YY}\bar{\ZZ}}\right\}=
\rank\left\{\miS^\circ_{00\bar{\Zsf}}\right\}=K$. This combined with the fact that $K < \max\{\bar{\Xsf}, \bar{\Ysf}, \bar{\Zsf}\}$ and $\miS^\circ_{\bar{\XX}\bar{\YY}\bar{\ZZ}}$ belongs to  $\mathcal{T}_{\bar{\Xsf}\bar{\Ysf}\bar{\Zsf}}$, by Lemma \ref{PrettyLemma}, $\miS^\circ_{\bar{\XX}\bar{\YY}\bar{\ZZ}}$ admits a unique Vandermonde decomposition of order $K$ from which the set of frequencies 
$\mif^\star_{1:K}$ can be uniquely determined. 

\vspace{-0.5cm}
\section{Proof of Corollary \ref{CorollaryTh1}}
\label{proofcorollary}
To prove  Corollary \ref{CorollaryTh1}  it is enough to prove 
the following lemma: 
\begin{lemma}
\label{Lemma:Corollary}
Under Assumption~\ref{A1}, any sensing matrix  $\miA$ 
belonging to the subset $\mathcal{A}^{(K)}_{\Csf}$ of the \well~sensing matrices set (see Definition \ref{def:wellstruct}), belongs to $\mathcal{A}^{(K)}$ as in Definition \ref{setinjective} provided that 
$S_{\Csf}  \geq 2K+(\dd-1)$.
\end{lemma}
\begin{proof}
Starting from Definition \ref{def:wellstruct}, the proof uses  \cite[ Th.~3]{Jiang01} and \cite[Lemma 1]{Jiang01} to prove conditions 1.C and 2.C respectively.
\end{proof}

\vspace{-0.5cm}
\section{Proof of Theorem~\ref{theoremApiProbalistic}}\label{appAIW0probabilistic}
 
Recall that in Theorem~\ref{theoremApi}, the condition for frequency recovery has been given assuming $\miA\in\mathcal{A}^{(K)}$. 
Hence, the proof of Theorem~\ref{theoremApiProbalistic} can be  identical to that given in Appendix~\ref{appAIW0} for Theorem~\ref{theoremApi} except that
in Theorem~\ref{theoremApiProbalistic},  we  will exploit the fact that $\mathcal{A}^{(K)}\equiv \mathcal{A}$ with high probability. Hence, to prove Theorem~\ref{theoremApiProbalistic} is it sufficient to prove the following Lemma.

\begin{lemma}
\label{Lemma:theoremApiProbalistic} 

 Let $\mathcal{A}$ be the sensing set defined in \eqref{eq:sensing_set}.
Then, for any binary sensing matrix $\miA \in \mathcal{A}$, and under Assumption~\ref{A1}, $\mathcal{A}^{(K)}\equiv \mathcal{A}$ with probability $1-\epsilon$ provided that $N\geq 2KC\log(2K\epsilon^{-1})$, with $C$ that is not larger than $ 12$.
\end{lemma}

To prove  Lemma \ref{Lemma:theoremApiProbalistic}, we need then to verify that conditions 1.C and 2.C in Definition \ref{setinjective} hold with probability $1-\epsilon$ for any $\miA \in \mathcal{A}$ provided that $N\geq 2KC\log(2K\epsilon^{-1})$, with $C$ that is not larger than $ 12$. Let us first state the following lemma,

\begin{lemma}
\label{LemmaB}
Consider a set of $L$ frequencies,  $\mif_{1:L}$,  satisfying Assumption \ref{A1}. For any $\miA \in \mathcal{A}$, and for any $p\times N$,
 random matrix $\miB=[\mib_1, \ldots, \mib_N]$ with $p \leq N$ and $p=O(N)$, such that $\mathbb{E} \big [\miB^\dagger\miB \big] =\miI_N$, $\mathbb{E} \big [b^*_{il} b_{jm}\big] =0$, $\forall i\neq j$ or $\forall l\neq m$ and 
 $\forall i\neq j$
\begin{eqnarray}
\mathbb{E} \big [b^*_{il} b_{jl} b_{im} b^*_{j\ell}\big] =
\left \{
\begin{array}{ll}
\mathbb{E} \big [|b_{il}|^2\big]
\mathbb{E} \big [ |b_{jl}|^2 \big]    & \mbox{if}\,  l=m=\ell \\
 0    & \mbox{otherwise} 
\end{array}
\right.,
\label{ufaaa}
\end{eqnarray}

we have with probability $1-\epsilon$ that:
	\begin{eqnarray}
	\delta_{L} = 
	\|  \miXi^\dagger \miXi -  \miI_{ L} \|_{2 \rightarrow 2} < 1
	\label{matrix_condition}
\end{eqnarray}
provided that $L< p$, and $N
\geq LC\log\left(L\epsilon^{-1}\right)$,
where $C$ is an appropriate universal constant,  $\miXi = c \miB\miA \miR_{\bar{\NN}}({\mif}_{1:L})\in\C^{p\times L}$ and $c$ is an arbitrary non-zero  constant.
\end{lemma}

\begin{proof}
Lemma \ref{LemmaB} proof uses the classical definition of $\|\cdot\|_{2\rightarrow 2}$ norm, and typical tools from concentration inequalities, in particular \cite[Prop. 8.16]{CSbook}.
\end{proof}

Note that when $\delta_{L}\in(0,1)$, then 
the largest and smallest singular values of $\miB\miA \miR_{\bar{\NN}}({\mif}_{1:L})$ satisfy
$\sigma_{\sf min}(\miB\miA \miR_{\bar{\NN}}({\mif}_{1:L})) \geq  \sqrt{1-\delta_{L}}/c$ and 
$\sigma_{\sf max}(\miB\miA \miR_{\bar{\NN}}({\mif}_{1:L})) \leq  \sqrt{1+\delta_{L}}/c$ which 
implies 
the injectivity of  $\miB\miA \miR_{\bar{\NN}}({\mif}_{1:L})$.  Note that if $\miB \miA \miR_{\bar{\NN}}({\mif}_{1:L})$ is injective and if $\miB$ is a  full rank $N \times N$ square matrix  then the $N \times L$ matrix  $\miA\miR_{\bar{\NN}}({\mif}_{1:L})$ is also injective \cite{Matsaglia74}.

To prove 1.C, let us consider a  set of $2K$ frequencies,  $\mif_{1:2K}$,  satisfying Assumption \ref{A1}. Setting  $L=2K$ and $\miB=\diag([d_1, \ldots, d_N])$ with $d_i$, $i\in[N]$ modelled as i.i.d. zero mean and unit variance random variables, we have that  the square diagonal matrix  $\miB$ is almost surely full rank. Furthermore, $\mathbb{E} \big [\miB^\dagger\miB \big] =\miI_{N}$, $\mathbb{E} \big [b^*_{il} b_{jm}\big] =0$, $\forall i\neq j$ or $\forall l\neq m$ while \eqref{ufaaa} trivially holds. Hence, from Lemma \ref{LemmaB}, it follows immediately that $\miB \miA \miR_{\bar{\NN}}({\mif}_{1:2K})$ is injective with  probability $1-\epsilon$ provided that $N
\geq LC\log\left(L\epsilon^{-1}\right)$ 
with  $C$  an  appropriate   universal constant not larger than $12$. Finally, due to the full rank property of the square matrix $\miB$,
the injectivity result also holds for  $\miA \miR_{\bar{\NN}}({\mif}_{1:2K})$, which concludes the proof of condition 1.C.


Let us now prove condition 2.C. To this end,
consider a  set of $K$ frequencies,  $\mif^\star_{1:K}$,  satisfying Assumption \ref{A1} and an arbitrary  set  
of $r^{\diamond} \leq K$  frequencies, $\mif^{\diamond}_{1:r^{\diamond}}$, such that for any  $k \in [K]$ and $j\in [r^{\diamond}]$,  $\mif^{\diamond}_{j}\neq \mif^\star_{k}$.
By Lemma \ref{LemmaB}, we can show that 
$\rank\left\{\miA \miR_{\bar{\NN}}({\mif}^\star_{1:K}) \} \right)=K$. Furthermore, denote by $r$ the rank of $\miA\miR_{\bar{\NN}}(\mif^{\diamond}_{1:r^{\diamond}})$, and, with no loss of generality, assume that the first $r$ columns of
$\miA\miR_{\bar{\NN}}(\mif^{\diamond}_{1:r^{\diamond}})$ are linearly independent and  denote such columns by    ${\mif}^{\diamond}_{1:r}$.
To prove condition 2.C. we have to prove that 
$\rank\left\{\miA \miR_{\bar{\NN}}({\mif}^{\star\diamond}_{1:K+r} ) \right\}= K +r$, where $\mif^{\star\diamond}_{1:K+r}=[\mif^{\star}_{1:K}
\mif^{\diamond}_{1:r}]$. To this end, it is enough to prove that if $\miQ$ is an arbitrary full $\rank$ $N \times N$ square matrix then 
$\miQ\miA\miR_{\bar{\NN}}(\mif^{\star\diamond}_{1:K+r})$ is also full $\rank$.  
Denoting by $\miX^\diamond=\miQ\miA\miR_{\bar{\NN}}(\mif^{\diamond}_{1:r})$, 
observing that  $\miX^\diamond$ is full $\rank$ and 
using the determinant of a block matrix,
we have that to prove condition 2.C it is sufficient to prove that $\det\{\miR^\dagger_{\bar{\NN}}(\mif^{\star}_{1:K})\miA^\dagger\miQ^\dagger(\miI_{ N}-\miX^\diamond(\miX^{\diamond\dagger} \miX^\diamond)^{-1}\miX^{\diamond\dagger})\miQ\miA\miR_{\bar{\NN}}(\mif^{\star}_{1:K})\}\neq 0$.

Next, let $\miU\boldsymbol{\Sigma}\miV^\dagger$ be the singular value decomposition of
$\miA\miR_{\bar{\NN}}(\mif^{\diamond}_{1:r})$ with 
 $\boldsymbol{\Sigma}$, $\miV^\dagger$  and  $\miU$ denoting the square diagonal  $r \times r$ 
singular-value matrix, the $r\times r$ right-singular eigenvector matrix and the  $N\times r$  left-singular eigenvector matrix of  $\miA\miR_{\bar{\NN}}(\mif^{\diamond}_{1:r})$.
Furthermore, set $\miQ=\diag (\mid)\miU_\mie^\dagger$
with $\mid=[d_1, \ldots, d_N]$ 
a possibly complex vector with i.i.d. zero mean and unit variance components and with $\miU_\mie=[\miU\, \miE]$ a unitary matrix whose first $r $ columns are the left-singular vectors of $\miA\miR_{\bar{\NN}}(\mif^{\diamond}_{1:r})$, $\miU$, while the remaining $N-r$ columns are uniformly distributed over the manifold $\miE^\dagger \miE= \miI_{N-r}$ such that
$\miE\miE^\dagger$ is a projector on the orthogonal complement to the subspace described by the columns of $\miU$, i.e., $\mbox{span}\{\miU\}$.

Based on the above definitions,  after some algebraic manipulations, we have that 
$$
\miQ^\dagger(\miI_{ N}-\miX^\diamond(\miX^{\diamond\dagger} \miX^\diamond)^{-1}\miX^{\diamond\dagger})\miQ= \miE 
\diag([|d_{r+1}|^2, \ldots  |d_{N}|^2 ])
\miE^\dagger
$$
Defining $\miB= \diag([d_{r+1}, \ldots  d_{N} ])
\miE^\dagger$,
if we set $L=K$ and $p=N-r$, since 
$K< N-r$, 
by Lemma \ref{LemmaB} we have that\footnote{After some algebraic manipulation, it can be easily shown that  $\mathbb{E} \big [\miB^\dagger\miB \big] =\miI_{ N}$, $\mathbb{E} \big [b^*_{il} b_{jm}\big] =0,$ $\forall i\neq j$ or $\forall l\neq m$ and  \eqref{ufaaa} holds.} $\miB
\miA\miR_{\bar{\NN}}(\mif^{\star}_{1:K})$ is injective and consequently 
$\det\{\miR_{\bar{\NN}}(\mif^{\star}
_{1:K})^\dagger  \miA^\dagger
\miB^\dagger \miB
\miA\miR_{\bar{\NN}}(\mif^{\star}_{1:K})\}\neq 0$ if $N
\geq LC\log\left(L\epsilon^{-1}\right)$ with  $C$  an  appropriate universal constant not larger than  $ 12 $. This completes the proof of 2.C and also Lemma \ref{Lemma:theoremApiProbalistic}.

\vspace{-0.75cm}
\bibliographystyle{ieeetr}
\bibliography{Bell-Labs_bib,mimo,compress-sensing,AoA}

\begin{thebibliography}{10}

\bibitem{Heath16}
R.~W. {Heath}, N.~{González-Prelcic}, S.~{Rangan}, W.~{Roh}, and A.~M.
  {Sayeed}, ``An overview of signal processing techniques for millimeter wave
  {MIMO} systems,'' {\em IEEE Journal of Selected Topics in Signal Processing},
  vol.~10, pp.~436--453, April 2016.

\bibitem{Capon69}
J.~Capon, ``High-resolution frequency-wavenumber spectrum analysis,'' {\em
  Proceedings of the IEEE}, vol.~57, pp.~1408--1418, Aug 1969.

\bibitem{Stoica89}
P.~Stoica and A.~Nehorai, ``{MUSIC}, maximum likelihood, and {Cramer-Rao}
  bound,'' {\em IEEE Transactions on Acoustics, Speech, and Signal Processing},
  vol.~37, pp.~720--741, May 1989.

\bibitem{Roy89}
R.~{Roy} and T.~{Kailath}, ``{ESPRIT}-estimation of signal parameters via
  rotational invariance techniques,'' {\em IEEE Transactions on Acoustics,
  Speech, and Signal Processing}, vol.~37, pp.~984--995, July 1989.

\bibitem{Candes06}
E.~J. {Candes}, J.~{Romberg}, and T.~{Tao}, ``Robust uncertainty principles:
  exact signal reconstruction from highly incomplete frequency information,''
  {\em IEEE Transactions on Information Theory}, vol.~52, pp.~489--509, Feb
  2006.

\bibitem{Stoica11}
P.~{Stoica}, P.~{Babu}, and J.~{Li}, ``Spice: A sparse covariance-based
  estimation method for array processing,'' {\em IEEE Transactions on Signal
  Processing}, vol.~59, pp.~629--638, Feb 2011.

\bibitem{Heckel18}
R.~Heckel and M.~Soltanolkotabi, ``Generalized line spectral estimation via
  convex optimization,'' {\em IEEE Transactions on Information Theory},
  vol.~64, pp.~4001--4023, June 2018.

\bibitem{Hu12}
L.~{Hu}, Z.~{Shi}, J.~{Zhou}, and Q.~{Fu}, ``Compressed sensing of complex
  sinusoids: An approach based on dictionary refinement,'' {\em IEEE
  Transactions on Signal Processing}, vol.~60, pp.~3809--3822, July 2012.

\bibitem{Yang13}
Z.~{Yang}, L.~{Xie}, and C.~{Zhang}, ``Off-grid direction of arrival estimation
  using sparse bayesian inference,'' {\em IEEE Transactions on Signal
  Processing}, vol.~61, pp.~38--43, Jan 2013.

\bibitem{Bhaskar13}
B.~N. Bhaskar, G.~Tang, and B.~Recht, ``Atomic norm denoising with applications
  to line spectral estimation,'' {\em IEEE Transactions on Signal Processing},
  vol.~61, pp.~5987--5999, Dec 2013.

\bibitem{Tang13}
G.~{Tang}, B.~N. {Bhaskar}, P.~{Shah}, and B.~{Recht}, ``Compressed sensing off
  the grid,'' {\em IEEE Transactions on Information Theory}, vol.~59,
  pp.~7465--7490, Nov 2013.

\bibitem{Tan14}
Z.~Tan, Y.~C. Eldar, and A.~Nehorai, ``Direction of arrival estimation using
  co-prime arrays: A super resolution viewpoint,'' {\em IEEE Transactions on
  Signal Processing}, vol.~62, pp.~5565--5576, Nov 2014.

\bibitem{Semper18}
S.~Semper, F.~Roemer, T.~Hotz, and G.~D. Galdo, ``Grid-free
  direction-of-arrival estimation with compressed sensing and arbitrary antenna
  arrays,'' in {\em 2018 IEEE International Conference on Acoustics, Speech and
  Signal Processing (ICASSP)}, pp.~3251--3255, April 2018.

\bibitem{Yang16}
Z.~Yang, L.~Xie, and P.~Stoica, ``Vandermonde decomposition of multilevel
  toeplitz matrices with application to multidimensional super-resolution,''
  {\em IEEE Transactions on Information Theory}, vol.~62, pp.~3685--3701, June
  2016.

\bibitem{Chu19}
H.~Chu, L.~Zheng, and X.~Wang, ``Super-resolution {mmWave} channel estimation
  using atomic norm minimization,'' {\em IEEE Journal of Selected Topics in
  Signal Processing}, vol.~13, pp.~1336--1347, October 2019.

\bibitem{Chi15}
Y.~Chi and Y.~Chen, ``Compressive two-dimensional harmonic retrieval via atomic
  norm minimization,'' {\em IEEE Transactions on Signal Processing}, vol.~63,
  pp.~1030--1042, Feb 2015.

\bibitem{Pan18}
Y.~N. Jian~Pan, Jun~Tang, ``Fast two-dimensional atomic norm minimization in
  spectrum estimation and denoising,'' {\em CoRR}, vol.~abs/1807.08606, 2018.

\bibitem{Zheng12}
J.~Zheng, X.~Gao, Z.~Zhang, and Z.~Feng, ``A compact eighteen-port antenna cube
  for {MIMO} systems,'' {\em IEEE Transactions on Antennas and Propagation},
  vol.~60, pp.~445--455, Feb 2012.

\bibitem{Wu17b}
N.~Wu, F.~Zhu, and Q.~Liang, ``Evaluating spatial resolution and channel
  capacity of sparse cylindrical arrays for massive {MIMO},'' {\em IEEE
  Access}, vol.~5, pp.~23994--24003, 2017.

\bibitem{Hu18}
S.~Hu, F.~Rusek, and O.~Edfors, ``Beyond massive {MIMO}: The potential of
  positioning with large intelligent surfaces,'' {\em IEEE Transactions on
  Signal Processing}, vol.~66, pp.~1761--1774, April 2018.

\bibitem{CSbook}
S.~Foucart and H.~Rauhut, {\em A mathematical introduction to compressive
  sensing}.
\newblock Springer, 2013.

\bibitem{Chandrasekaran2012}
V.~Chandrasekaran, B.~Recht, P.~A. Parrilo, and A.~S. Willsky, ``The convex
  geometry of linear inverse problems,'' {\em Foundations of Computational
  Mathematics}, vol.~12, no.~6, pp.~805--849, 2012.

\bibitem{Krishnan2005}
K.~Krishnan and T.~Terlaky, {\em Interior Point and Semidefinite Approaches in
  Combinatorial Optimization}, pp.~101--157.
\newblock Boston, MA: Springer US, 2005.

\bibitem{Gurvits02}
L.~Gurvits and H.~Barnum, ``Largest separable balls around the maximally mixed
  bipartite quantum state,'' {\em Physical Review A}, vol.~66, 04 2002.

\bibitem{Jiang01}
{Tao Jiang}, N.~D. {Sidiropoulos}, and J.~M.~F. {ten Berge}, ``Almost-sure
  identifiability of multidimensional harmonic retrieval,'' {\em IEEE
  Transactions on Signal Processing}, vol.~49, pp.~1849--1859, Sep. 2001.

\bibitem{Matsaglia74}
G.~Matsaglia and G.~P.~H. Styan, ``Equalities and inequalities for ranks of
  matrices,'' {\em Linear and Multilinear Algebra}, vol.~2, no.~3,
  pp.~269--292, 1974.

\end{thebibliography}

\end{document}